\documentclass[11pt]{article}
\usepackage[margin=0.75in]{geometry}
\pdfoutput=1
\usepackage[utf8]{inputenc}
\usepackage[english]{babel}
\usepackage[T1]{fontenc}
\usepackage{amsmath,amssymb}
\usepackage{caption}
\usepackage{authblk}
\usepackage{subcaption}
\usepackage{graphicx}           
\usepackage{hyperref}
\usepackage{relsize}
\usepackage{tcolorbox}
\usepackage{thm-restate}
\usepackage{cleveref}
\hypersetup{
    colorlinks=true,
    linkcolor=blue,
    citecolor=magenta,
    filecolor=magenta,      
    urlcolor=blue,
    pdftitle={},
    pdfpagemode=FullScreen,
    }
\usepackage{soul}
\usepackage{mathtools}

\usepackage{tikz}
\usetikzlibrary{shapes}
\usepackage{xcolor}
\usepackage{amsthm}
\theoremstyle{definition}
\usepackage{float}
\makeatletter
\newlength\min@xx

\makeatother

\usepackage{lipsum}

\newtheorem{theorem}{Theorem}
\newtheorem{lemma}{Lemma}

\newtheorem{definition}{Definition}

\newtheorem{proposition}{Proposition}
\DeclareMathOperator{\Tr}{Tr}

\newcommand{\hexagon}[3]{
\begin{scope}[xshift=#2 *1cm -0.5cm *#3, yshift=#3*0.866cm]
\foreach \t in {30, 90, 150, 210, 270, 330}
{
\draw[] (\t:#1*0.577) -- (\t+60:#1*0.577);
}
\end{scope}
}

\newcommand{\hexagonunitcolor}[4]{
\begin{scope}[xshift=#2 *1cm -0.5cm *#3, yshift=#3*0.866cm]
\foreach \t in {30, 90, 150, 210, 270, 330}
{
\draw[thick] (\t:#1*0.577) -- (\t+60:#1*0.577);
}
\end{scope}
\begin{scope}[xshift=#2 *1cm -0.5cm *#3, yshift=#3*0.866cm]
\draw[fill=#4, opacity=0.2] (30:#1*0.577) -- (90:#1*0.577) -- (150:#1*0.577) -- (210:#1*0.577) -- (270:#1*0.577) -- (330:#1*0.577) -- cycle;
\end{scope}
}

\definecolor{BSorange}{RGB}{140,50,0}

\begin{document}

\title{Strict area law implies commuting parent Hamiltonian}
\author[1]{Isaac H.~Kim}
\affil[1]{Department of Computer Science, University of California, Davis, CA 95616, USA}
\author[2]{Ting-Chun Lin}
\affil[2]{Department of Physics, University of California San Diego, La Jolla, CA 92093, USA}
\author[3]{Daniel Ranard}
\affil[3]{Center for Theoretical Physics, Massachusetts Institute of Technology, Cambridge, MA 02139}
\author[2]{Bowen Shi}

\date{}
\maketitle

\begin{abstract}
We show that in two spatial dimensions, when a quantum state has entanglement entropy obeying a strict area law, meaning $S(A)=\alpha |\partial A| - \gamma$ for constants $\alpha, \gamma$ independent of lattice region 
$A$, then it admits a commuting parent Hamiltonian.  More generally, we prove that the entanglement bootstrap axioms in 2D imply the existence of a commuting, local parent Hamiltonian with a stable spectral gap. We also extend our proof to states that describe gapped domain walls. Physically, these results imply that the states studied in the entanglement bootstrap program correspond to ground states of some local Hamiltonian, describing a stable phase of matter. Our result also suggests that systems with chiral gapless edge modes cannot obey a strict area law provided they have finite local Hilbert space.
\end{abstract}

\section{Introduction}
\label{sec:intro}

A fundamental question in the study of quantum many-body systems is the classification of phases. Formally, one can define gapped phases of matter by defining an equivalence relation in the space of Hamiltonians. If two ground states are adiabatically connected to each other, the two systems are said to be in the same phase. 

The classification of gapped quantum many-body systems has been completed in one spatial dimensions~\cite{Chen2011,Schuch2011,Fidkowski2011}. However, in higher dimensions, much less is known. In two spatial dimensions, a common expectation is that each gapped phase of matter is uniquely labeled by the combination of some unitary modular tensor category and a rational number known as the chiral central charge~\cite{kitaev2006anyons}. Though this expectation seems plausible, establishing this fact from a microscopic point of view has been challenging. In three dimensions and higher, there is not even an exhaustive proposed classification, at least if one attempts to accommodate exotic models such as fractons~\cite{Haah2011}.

Recently, a new paradigm to tackle this problem was developed: the entanglement bootstrap~\cite{shi2020fusion}. This is an approach to derive the universal properties of gapped quantum phases from ground state entanglement, and ideally to classify them. In entanglement bootstrap, one begins with a set of reasonable axioms about the ground state entanglement, which are believed to be approximately satisfied for gapped ground states. One sufficient axiom is to assume a ``strict area law'' on the entanglement entropy $S(A)$ of region $A$, \cite{Kitaev2006,Levin2006}
\begin{align} \label{eq:strict-area}
S(A)=\alpha |\partial A| - \gamma
\end{align}
for constants $\alpha, \gamma$ independent of $A$.  While this assumption is particularly compact, weaker axioms involving linear combinations of entropies are sufficient and perhaps more natural, described in Section \ref{subsec:summary-of-results}.

While the entanglement bootstrap axioms are not \textit{precisely} satisfied in most natural systems, they provide a simple framework in which otherwise difficult ideas can be explored rigorously. Early developments of the entanglement bootstrap program focused on proving familiar ideas about gapped systems in two spatial dimensions, such as the anyon fusion rules~\cite{shi2020fusion}. More recent progress in this direction led to novel superselection sectors and their fusion rules in gapped domain walls~\cite{Shi2021}, exotic topological charges in three spatial dimensions~\cite{Huang2023,Shi2023}, as well as an expression for chiral central charge~\cite{Kim2022,Kim2022a}. There is even a possible extension of this approach to critical systems ~\cite{Lin2023}.

While these past works suggest strongly that the entanglement bootstrap approach can lead to a novel predictions about the universal properties of gapped ground states, there is a potential concern. While the axioms of the entanglement bootstrap are known to hold for a large class of gapped ground states, one may wonder if the converse is true. Namely, given some state that satisfies the axioms, one might ask if this state can be viewed as a ground state of some local Hamiltonian. If such a parent Hamiltonian does not exist, not all the findings made in the entanglement bootstrap program may be physically relevant, because then the universal properties revealed by the purported ``ground state'' might not belong to any local Hamiltonian.

To that end, we prove rigorously that any state that satisfies the axioms of the entanglement bootstrap exactly has a gapped parent Hamiltonian with a stable spectral gap. More specifically, we show that the axioms in Ref.~\cite{shi2020fusion} --- the ones that are applicable to the ``bulk'' of gapped ground states in 2D without any domain wall --- imply that there is a parent Hamiltonian with a spectral gap. In fact, this parent Hamiltonian is a sum of local terms each of which are commuting projectors. Moreover, we also show that this Hamiltonian satisfies a set of conditions that guarantees the stability of the gap~\cite{Michalakis2013}. Therefore, any state that satisfies the axioms of the entanglement bootstrap can be viewed as a ground state of some local Hamiltonian, which is a representative of a stable phase of matter that is robust against sufficiently weak (but nonvanishing) perturbation.

The same argument also works in the presence of gapped domain walls~\cite{Shi2021}. This is a setup in which some of the axioms are relaxed along a codimension-$1$ defect, known as a \emph{gapped domain wall}. The axioms for domain walls in Ref.~\cite{Shi2021} were motivated by considering entanglement properties (such as the strict area law), independent of other inputs such as topological quantum field theory. Therefore, the existence of the gapped parent Hamiltonian in this context was not a priori obvious.

The main insight behind our work is an observation about the property of a quantum Markov chain~\cite{Petz1988}, which may be of independent interest. A tripartite quantum state $\rho_{ABC}$ is said to be a quantum Markov chain if it satisfies $S(\rho_{AB}) + S(\rho_{BC}) = S(\rho_B) + S(\rho_{ABC})$, where $S(\rho)=-\text{Tr}(\rho \log \rho)$ is the von Neumann entropy of a density matrix $\rho$. We prove that the projector onto the support of $\rho_{AB}$ and $\rho_{BC}$ must necessarily commute. This fact, together with the known properties of quantum Markov chains used in the entanglement bootstrap program~\cite{shi2020fusion}, lets us construct a local parent Hamiltonian consisting of commuting projectors.

Our result has two immediate physical implications. First, it implies that the states that satisfy the axioms of the entanglement bootstrap ought to be a representative ground state wavefunction of some stable phase of matter. This is because there is a local parent Hamiltonian with a stable spectral gap which has this wavefunction as a ground state. Second, our result implies that there is an obstruction to satisfying the axioms of the entanglement bootstrap \emph{exactly} for systems with a nonzero chiral central charge. It is well-known that such systems must host a nonzero energy current along the edge at finite temperature~\cite{kitaev2006anyons}. If the axioms are satisfied exactly, by our result there must exist a local commuting parent Hamiltonian. For such Hamiltonian, any energy current is identically zero at all temperature, a direct contradiction with the fact that nonzero chiral central charge implies a nonzero energy current at finite temperature.

The rest of the paper is structured as follows. In Section~\ref{sec:summary} we discuss the setup and provide an executive summary of the main results. In Section~\ref{sec:preliminary} we discuss well-known facts that are relevant to this paper. In Section~\ref{sec:extensions_of_axioms}, we review the argument for extending the axioms of the entanglement bootstrap~\cite{shi2020fusion}, focusing on the setup we advocated in Section~\ref{sec:summary}. In Section~\ref{sec:ltqo}, we  introduce a class of parent Hamiltonians and show that they all satisfy a property known as the local topological quantum order condition~\cite{Michalakis2013}. In Section~\ref{sec:commuting_projectors}, we prove the aforementioned property of quantum Markov chain: that projectors onto the support over the reduced density matrices commute with each other. In Section~\ref{sec:parent_hamiltonian_final}, we construct the parent Hamiltonian that has a stable spectral gap. In Section~\ref{sec:gapped_domain_wall}, we extend the construction to the case in which a gapped domain wall is present. We note that the Hamiltonians constructed in these sections are quite elaborate, consisting of local terms acting on more than hundred sites. In Section~\ref{sec:weight_reduction}, we introduce a technique to reduce this number to three, though we no longer have a guarantee that the local terms commute with each other. We conclude with a discussion and open problems in Section~\ref{sec:discussion}.

\section{Results and strategy}
\label{sec:summary}

\subsection{Informal overview}
\label{sec:informal_overview}

We begin with an informal sketch.  Suppose a 2D state $\rho$ satisfies 
\begin{align} \label{eq:zero_CMI}
    I(A:C|B)_{\rho}=0 
\end{align}
for all regions $ABC$ topologically equivalent to those of [Fig.~\ref{fig:chain_like_ex}]. Here, $I(A:C|B):= S(AB)+S(BC)-S(B)-S(ABC)$ denotes the conditional mutual information. Note this property is true for states with strict area law, as in Eq.~\eqref{eq:strict-area}, or more generally for states that satisfy the entanglement bootstrap axioms [Section \ref{subsec:summary-of-results}]. Our main result is the construction of a commuting parent Hamiltonian $H$ for which $\rho$ is a gapped and frustration-free ground state, with a local uniqueness property described below.

\begin{figure}[h]
    \centering
    \begin{tikzpicture}[scale=0.8]
        \draw[] (0,0) -- ++ (1,0) -- ++ (0, 1) -- ++ (-1, 0) -- cycle;
        \draw[] (1,0) -- ++ (1,0) -- ++ (0, 1) -- ++ (-1, 0) -- cycle;
        \draw[] (2,0) -- ++ (1,0) -- ++ (0, 1) -- ++ (-1, 0) -- cycle;
        \node[] () at (0.5, 0.5) {$A$};
        \node[] () at (1.5, 0.5) {$B$};
        \node[] () at (2.5, 0.5) {$C$};
    \end{tikzpicture}
    \caption{Sub-regions $A,B$ and $C$ in the 2D plane arranged like a chain.}
    \label{fig:chain_like_ex}
\end{figure}
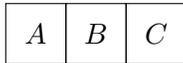

The construction is easily summarized as follows. For any region $X$, let 
\begin{align}
P_X=\text{proj}_{\text{supp}(\rho_X)}
\end{align} 
be the projector onto the support of $\rho_X$. For states $\rho_{ABC}$ on a finite-dimensional Hilbert space satisfying $I(A:C|B)_{\rho}=0$, it turns out [Proposition~\ref{proposition:commutation_markov}]
\begin{align} \label{eq:P-ABC-commute}
    [P_{AB},P_{BC}] = 0.
\end{align}
We will construct commuting parent Hamiltonian for $\rho$ by taking
\begin{align}
    H = \sum_{X\in \mathcal{S}} (1-P_X)
\end{align}
where the sum is taken over some set $\mathcal{S}$ of specially chosen $O(1)$-size regions $X$. We guarantee the terms $P_X$ commute by ensuring that whenever a pair of overlapping regions $X_1, X_2 \in \mathcal{S}$ appear in the sum, they are topologically equivalent to the pair of regions $AB, BC$ [Fig.~\ref{fig:chain_like_ex}], so that Eq.~\eqref{eq:P-ABC-commute} guarantees $[P_{X_1},P_{X_2}]=0$. 

We also want $H$ to have a unique ground state in some sense -- for instance, taking $H=0$ gives a ``parent Hamiltonian,'' but it is not useful. Still, we do not want to require $H$ has unique ground state: models such as the toric code have degeneracies depending on the global topology of the manifold.  A natural condition is that the ground state is ``locally unique,'' in the sense that all zero-energy states on a disk look identical on the interior of the disk.  One codification of this is the condition of \emph{local topological quantum order} (LTQO)~\cite{Michalakis2013}.
To guarantee LTQO,
we require the regions in $\mathcal{S}$ form a sort of  
``open cover'' of the plane, i.e., their interiors must cover the plane.  The LTQO property then follows from basic entanglement bootstrap techniques.

To find a set of regions satisfying these geometric requirements, we make use of a cellular decomposition of the plane. In particular, referring to Fig.~\ref{fig:coarse_grained_v2}, each region $X$ is given by the union of a single green region with its neighboring blue and red regions.  One can check such regions $X$ satisfy the above conditions.

Perhaps one can see that the basic ideas are topological and not too intricate.  However, some care is required in the choice of regions $X$.  Much of the technical detail is devoted to making a rigorous argument while assuming the entanglement bootstrap axioms only for a small number of explicitly chosen, $O(1)$-size lattice regions and their translations.

\subsection{Summary of results} \label{subsec:summary-of-results}

Throughout this paper, we shall consider a two-dimensional (2D) plane, covered by a set of faces, each being a closed hexagon-shaped disk [Fig.~\ref{fig:threecolor}]. We assign a finite-dimensional Hilbert space to each face, and we informally view the global Hilbert space as a tensor product of the local Hilbert spaces. We choose a quantum state in the global Hilbert space, which we call the \emph{reference state} and denoted as $\sigma$ throughout the paper.\footnote{We treat the infinite lattice informally. More formally, one could adopt the algebraic framework \cite{naaijkens2013quantum}, ultimately viewing the reference state as a positive linear functional on the algebra of quasi-local observables.  Rather than assuming this ``global state'' as an input, one could first view the reference state as a collection of consistent density matrices on $O(1)$-size regions covering the plane, satisfying the entanglement bootstrap axioms locally.  Then the techniques we discuss allow one to define a consistent ``global state,'' e.g.~a state on the algebra of quasi-local observables, consistent with the local data.}

\begin{figure}[!h]
    \centering
    \begin{tikzpicture}[scale=0.6]
    \foreach \x in {0,...,8}
    {
    \foreach \y in {0,...,8}
    {
    \hexagon{1}{\x}{\y};
    }
    }
    \foreach \x in {0,1,2}
    {
    \foreach \y in {0,1,2}
    {
    \hexagonunitcolor{1}{3*\x}{\y*3}{white};
    \begin{scope}[xshift=1cm]
    \hexagonunitcolor{1}{3*\x}{\y*3}{white};
    \end{scope}
    \begin{scope}[xshift=2cm]
    \hexagonunitcolor{1}{3*\x}{\y*3}{white};
    \end{scope}
    }
    }

    \foreach \x in {0,1,2}
    {
    \foreach \y in {0,1,2}
    {
    \hexagonunitcolor{1}{3*\x}{(\y+0.333333)*3}{white};
    \begin{scope}[xshift=1cm]
    \hexagonunitcolor{1}{3*\x}{(\y+0.333333)*3}{white};
    \end{scope}
    \begin{scope}[xshift=2cm]
    \hexagonunitcolor{1}{3*\x}{(\y+0.333333)*3}{white};
    \end{scope}
    }
    }

    \foreach \x in {0,1,2}
    {
    \foreach \y in {0,1,2}
    {
    \hexagonunitcolor{1}{3*\x}{(\y+0.666666)*3}{white};
    \begin{scope}[xshift=1cm]
    \hexagonunitcolor{1}{3*\x}{(\y+0.666666)*3}{white};
    \end{scope}
    \begin{scope}[xshift=2cm]
    \hexagonunitcolor{1}{3*\x}{(\y+0.666666)*3}{white};
    \end{scope}
    }
    }
    \draw[dotted, line width=2pt] (7, 0.866cm*4) -- ++ (0.65cm, 0);
    \draw[dotted, line width=2pt] (-3, 0.866cm*4) -- ++ (-0.65cm, 0);
    \draw[dotted, line width=2pt] (4, -0.866cm) -- ++ (0.25cm*1.25, -0.433cm*1.25);
    \draw[dotted, line width=2pt] (0, 0.866cm*9) -- ++ (-0.25cm*1.25, 0.433cm*1.25);
    \end{tikzpicture}
    \caption{In this paper, we consider a partition of $\mathbb{R}^2$ into hexagonal cells, each representing a finite-dimensional quantum system. }
    \label{fig:threecolor}
\end{figure}
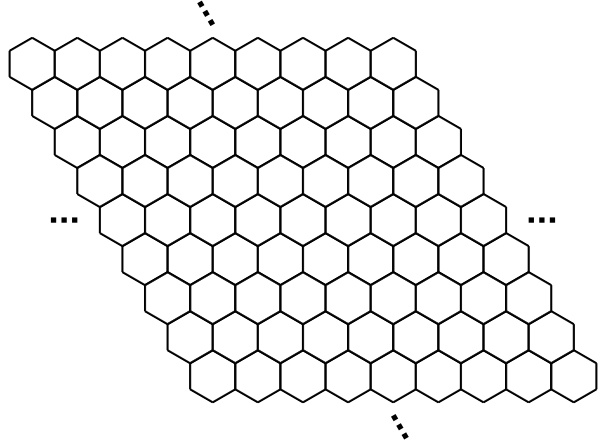

Denoting the set of faces as $\Lambda$, we define several notions that shall be useful. For a set $A\subset \Lambda$, we denote the \emph{neighborhood} of $A$ as $N(A)$, which is a set of faces in $\Lambda \setminus A$ that are adjacent to at least one of the faces in $A$. We shall sometimes focus on the boundary of a single face $f\in \Lambda$. In those cases, we shall abuse the notation and denote the neighborhood as $N(f)$, instead of $N(\{f \})$. We say that a set $A\subset \Lambda$ is \emph{connected} if the union of the faces, viewed as a subset of $\mathbb{R}^2$, is connected. Similarly, $A\subset \Lambda$ is \emph{simply connected} if the union of the faces is, viewed as a subset of $\mathbb{R}^2$, is simply connected. Examples are shown in Fig.~\ref{fig:disk-like}. Sometimes we shall refer to a simply connected set as a \emph{disk}. Lastly, if a disk is $f \cup N(f)$ for some $f\in \Lambda$, we refer to this disk as an \emph{elementary disk}. We emphasize that elementary disk is not necessarily the smallest disk. Nevertheless, we advocate this convention due to the repeated appearance of elementary disk in this paper.

\begin{figure}[h]
    \centering
    \begin{tikzpicture}[scale=0.6]
        \hexagonunitcolor{1}{0}{0}{green};
        \hexagonunitcolor{1}{1}{0}{green};
        \hexagonunitcolor{1}{1}{1}{green};
        \hexagonunitcolor{1}{0}{1}{green};
        \hexagonunitcolor{1}{2}{1}{green};
        \hexagonunitcolor{1}{1}{2}{green};
        \hexagonunitcolor{1}{2}{2}{green};

        \begin{scope}[xshift=15cm]
        \hexagonunitcolor{1}{0}{0}{blue};
        \hexagonunitcolor{1}{1}{0}{blue};
        \hexagonunitcolor{1}{0}{1}{blue};
        \hexagonunitcolor{1}{2}{1}{blue};
        \hexagonunitcolor{1}{1}{2}{blue};
        \hexagonunitcolor{1}{2}{2}{blue};
        \end{scope}
    \end{tikzpicture}
    \caption{(Left) The faces colored in green is an example of a simply connected subsystem. (Right) The faces colored in blue is an example of a connected subsystem which is not simply connected.}
    \label{fig:disk-like}
\end{figure}
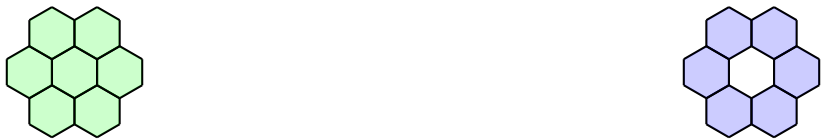

Our first main result concerns a reference state that satisfies the axioms \textbf{A0} and \textbf{A1} in the vicinity of a face. More formally, consider a disk $A\subset \Lambda$. 
\begin{definition}
    The state $\sigma$ satisfies \textbf{A0} in the vicinity of $A\subset V$ if 
    \begin{equation}
        \left(S(A) + S(A\cup N(A)) - S(N(A))\right)_{\sigma}=0.
    \end{equation}
\end{definition}
\begin{definition}
    The state $\sigma$ satisfies \textbf{A1} in the vicinity of $A\subset V$ if for any partition of $N(A)$ into disjoint simply connected subsets $N(A)_1$ and $N(A)_2$, 
    \begin{equation}
        \left(S(A\cup N(A)_1) - S(N(A)_1) + S(A\cup N(A)_2) - S(N(A)_2)\right)_{\sigma} = 0.
    \end{equation}
\end{definition}
\noindent 
We say that a reference state is \emph{valid} if for every face $f\in \Lambda$, \textbf{A0} and \textbf{A1} in the vicinity of $f$ is satisfied. To be more precise, we demand
\begin{equation}
\left(S(C) + S(C|B) \right)_{\sigma}=0 \quad \text{ for } \quad
    \begin{tikzpicture}[scale=0.45, baseline={([yshift=-.55ex]current bounding box.center)}]
        \hexagonunitcolor{1}{1}{1}{red};
        \hexagonunitcolor{1}{0}{0}{blue};
        \hexagonunitcolor{1}{1}{0}{blue};
        \hexagonunitcolor{1}{0}{1}{blue};
        \hexagonunitcolor{1}{2}{1}{blue};
        \hexagonunitcolor{1}{1}{2}{blue};
        \hexagonunitcolor{1}{2}{2}{blue};
        \node[] () at (0.5, 0.866) {$C$};
        \node[] () at (-0.5, 0.866) {$B$};
    \end{tikzpicture}
    \label{eq:a0_starting_point}
\end{equation}
where $S(X|Y) := S_{XY} - S_Y$ and
\begin{equation}
    \left(S(C|B) + S(C|D) \right)_{\sigma}=0 \quad \text{ for }\quad 
    \begin{tikzpicture}[scale=0.45, baseline={([yshift=-.55ex]current bounding box.center)}]
        \hexagonunitcolor{1}{1}{1}{red};
        \hexagonunitcolor{1}{0}{0}{green};
        \hexagonunitcolor{1}{1}{0}{green};
        \hexagonunitcolor{1}{0}{1}{blue};
        \hexagonunitcolor{1}{2}{1}{green};
        \hexagonunitcolor{1}{1}{2}{green};
        \hexagonunitcolor{1}{2}{2}{green};
        \node[] () at (0.5, 0.866) {$C$};
        \node[] () at (-0.5, 0.866) {$B$};
        \node[] () at (1.5, 0.866) {$D$};
    \end{tikzpicture},
    \begin{tikzpicture}[scale=0.45, baseline={([yshift=-.55ex]current bounding box.center)}]
        \hexagonunitcolor{1}{1}{1}{red};
        \hexagonunitcolor{1}{0}{0}{blue};
        \hexagonunitcolor{1}{1}{0}{green};
        \hexagonunitcolor{1}{0}{1}{blue};
        \hexagonunitcolor{1}{2}{1}{green};
        \hexagonunitcolor{1}{1}{2}{green};
        \hexagonunitcolor{1}{2}{2}{green};
        \node[] () at (0.5, 0.866) {$C$};
        \node[] () at (-0.5, 0.866) {$B$};
        \node[] () at (1.5, 0.866) {$D$};
    \end{tikzpicture},
    \begin{tikzpicture}[scale=0.45, baseline={([yshift=-.55ex]current bounding box.center)}]
        \hexagonunitcolor{1}{1}{1}{red};
        \hexagonunitcolor{1}{0}{0}{blue};
        \hexagonunitcolor{1}{1}{0}{blue};
        \hexagonunitcolor{1}{0}{1}{blue};
        \hexagonunitcolor{1}{2}{1}{green};
        \hexagonunitcolor{1}{1}{2}{green};
        \hexagonunitcolor{1}{2}{2}{green};
        \node[] () at (0.5, 0.866) {$C$};
        \node[] () at (-0.5, 0.866) {$B$};
        \node[] () at (1.5, 0.866) {$D$};
    \end{tikzpicture},
    \label{eq:a1_starting_point}
\end{equation}
as well as rotated choices of subsystems $B$ and $D$ (while keeping $C$ intact). 

Our first main result establishes, for the valid reference states, the existence of a local commuting projector parent Hamiltonian. These are Hamiltonians that can be expressed as a sum of local projectors that commute with each other. Here is our first main theorem.
\begin{theorem}
For any valid reference state, there exists a local commuting projector parent Hamiltonian satisfying the local topological quantum order condition~\cite[Definition 4]{Michalakis2013}.  
\label{theorem:main1}
\end{theorem}
\noindent 
In Ref.~\cite{Michalakis2013}, the local topological quantum order condition, together with a condition they refer to as the local gap condition, guarantees the stability of the gap under any perturbation which has a finite but sufficiently small strength.\footnote{The condition we use is actually TQO-2 in Ref.~\cite{Bravyi2010}, though we use the local TQO condition for convenience.} It should be noted that the local gap condition is trivially satisfied by any local commuting projector models. Therefore, an immediate corollary of Theorem~\ref{theorem:main1} is the existence of a parent Hamiltonian for any valid reference state, which has a stable spectral gap. 

Our second main result establishes the existence of a local commuting projector parent Hamiltonian for a larger class of states. In Ref.~\cite{Shi2021}, certain axioms of the entanglement bootstrap were relaxed along a one-dimensional line. As an example, we can consider a lattice shown in Fig.~\ref{fig:gdw_setup}. The reference state, in this case, will continue to satisfy \textbf{A0} everywhere, but \textbf{A1} may be violated along the disks that touch the red line in Fig.~\ref{fig:gdw_setup}. More precisely, while \textbf{A1} continues to be satisfied for the following choice of $B, C,$ and $D$,
\begin{equation}
    \begin{tikzpicture}[scale=0.45, baseline={([yshift=-.55ex]current bounding box.center)}]
        \hexagonunitcolor{1}{1}{1}{red};
        \hexagonunitcolor{1}{0}{0}{green};
        \hexagonunitcolor{1}{1}{0}{green};
        \hexagonunitcolor{1}{0}{1}{blue};
        \hexagonunitcolor{1}{2}{1}{green};
        \hexagonunitcolor{1}{1}{2}{green};
        \hexagonunitcolor{1}{2}{2}{green};
        \begin{scope}[xshift=0 * 1cm]
        \draw[red, line width=2pt] (150:0.577) -- (90:0.577) -- (30:0.577);
        \end{scope}
        \begin{scope}[xshift=1 * 1cm]
        \draw[red, line width=2pt] (150:0.577) -- (90:0.577) -- (30:0.577);
        \end{scope}
        \node[] () at (0.5, 0.866) {$C$};
        \node[] () at (-0.5, 0.866) {$B$};
        \node[] () at (1.5, 0.866) {$D$};
    \end{tikzpicture},
    \begin{tikzpicture}[scale=0.45, baseline={([yshift=-.55ex]current bounding box.center)}]
        \hexagonunitcolor{1}{1}{1}{red};
        \hexagonunitcolor{1}{0}{0}{blue};
        \hexagonunitcolor{1}{1}{0}{green};
        \hexagonunitcolor{1}{0}{1}{blue};
        \hexagonunitcolor{1}{2}{1}{green};
        \hexagonunitcolor{1}{1}{2}{green};
        \hexagonunitcolor{1}{2}{2}{green};
        \begin{scope}[xshift=0 * 1cm]
        \draw[red, line width=2pt] (150:0.577) -- (90:0.577) -- (30:0.577);
        \end{scope}
        \begin{scope}[xshift=1 * 1cm]
        \draw[red, line width=2pt] (150:0.577) -- (90:0.577) -- (30:0.577);
        \end{scope}
        \node[] () at (0.5, 0.866) {$C$};
        \node[] () at (-0.5, 0.866) {$B$};
        \node[] () at (1.5, 0.866) {$D$};
    \end{tikzpicture},
    \begin{tikzpicture}[scale=0.45, baseline={([yshift=-.55ex]current bounding box.center)}]
        \hexagonunitcolor{1}{1}{1}{red};
        \hexagonunitcolor{1}{0}{0}{blue};
        \hexagonunitcolor{1}{1}{0}{blue};
        \hexagonunitcolor{1}{0}{1}{blue};
        \hexagonunitcolor{1}{2}{1}{green};
        \hexagonunitcolor{1}{1}{2}{green};
        \hexagonunitcolor{1}{2}{2}{green};
        \begin{scope}[xshift=0 * 1cm]
        \draw[red, line width=2pt] (150:0.577) -- (90:0.577) -- (30:0.577);
        \end{scope}
        \begin{scope}[xshift=1 * 1cm]
        \draw[red, line width=2pt] (150:0.577) -- (90:0.577) -- (30:0.577);
        \end{scope}
        \node[] () at (0.5, 0.866) {$C$};
        \node[] () at (-0.5, 0.866) {$B$};
        \node[] () at (1.5, 0.866) {$D$};
    \end{tikzpicture},
    \label{eq:a1_boundary_valid}
\end{equation}
it may not necessarily hold for the following choices:
\begin{equation}
    \begin{tikzpicture}[scale=0.45, baseline={([yshift=-.55ex]current bounding box.center)}]
        \hexagonunitcolor{1}{1}{1}{red};
        \hexagonunitcolor{1}{0}{0}{green};
        \hexagonunitcolor{1}{1}{0}{green};
        \hexagonunitcolor{1}{0}{1}{green};
        \hexagonunitcolor{1}{2}{1}{green};
        \hexagonunitcolor{1}{1}{2}{blue};
        \hexagonunitcolor{1}{2}{2}{green};
        \begin{scope}[xshift=0 * 1cm]
        \draw[red, line width=2pt] (150:0.577) -- (90:0.577) -- (30:0.577);
        \end{scope}
        \begin{scope}[xshift=1 * 1cm]
        \draw[red, line width=2pt] (150:0.577) -- (90:0.577) -- (30:0.577);
        \end{scope}
        \node[] () at (0.5, 0.866) {$C$};
        \node[] () at (0, 1.732) {$B$};
        \node[] () at (1.5, 0.866) {$D$};
    \end{tikzpicture},
    \begin{tikzpicture}[scale=0.45, baseline={([yshift=-.55ex]current bounding box.center)}]
        \hexagonunitcolor{1}{1}{1}{red};
        \hexagonunitcolor{1}{0}{0}{green};
        \hexagonunitcolor{1}{1}{0}{green};
        \hexagonunitcolor{1}{0}{1}{green};
        \hexagonunitcolor{1}{2}{1}{green};
        \hexagonunitcolor{1}{1}{2}{blue};
        \hexagonunitcolor{1}{2}{2}{blue};
        \begin{scope}[xshift=0 * 1cm]
        \draw[red, line width=2pt] (150:0.577) -- (90:0.577) -- (30:0.577);
        \end{scope}
        \begin{scope}[xshift=1 * 1cm]
        \draw[red, line width=2pt] (150:0.577) -- (90:0.577) -- (30:0.577);
        \end{scope}
        \node[] () at (0.5, 0.866) {$C$};
        \node[] () at (0, 1.732) {$B$};
        \node[] () at (1.5, 0.866) {$D$};
    \end{tikzpicture}.
    \label{eq:a1_boundary_invalid}
\end{equation}
In words, \textbf{A1} may not hold if the boundaries between $B$ and $D$ are both away from the domain wall and lie on the same side with respect to the domain wall. Otherwise, the \textbf{A1} continues to hold. Again, we say that the reference state is a \emph{valid} reference state (for gapped domain walls) if it satisfies these constraints.  
\begin{theorem}
    \label{thm:main_result_2}
    For any valid reference state for gapped domain walls, there exists a local commuting projector parent Hamiltonian satisfying the local topological order condition.
\end{theorem}

\begin{figure}[t]
    \centering
    \begin{tikzpicture}[scale=0.6]
    \foreach \x in {0,...,8}
    {
    \foreach \y in {0,...,8}
    {
    \hexagon{1}{\x}{\y};
    }
    }
    \begin{scope}[xshift = -2cm, yshift= 4 * 0.866cm]
    \foreach \x in {0, ..., 8}
    {
    \begin{scope}[xshift=\x * 1cm]
        \draw[red, line width=2pt] (210:0.577) -- (270:0.577) -- (330:0.577);
    \end{scope}
    }
    \end{scope}
    
    \draw[dotted, line width=2pt] (7, 0.866cm*4) -- ++ (0.65cm, 0);
    \draw[dotted, line width=2pt] (-3, 0.866cm*4) -- ++ (-0.65cm, 0);
    \draw[dotted, line width=2pt] (4, -0.866cm) -- ++ (0.25cm*1.25, -0.433cm*1.25);
    \draw[dotted, line width=2pt] (0, 0.866cm*9) -- ++ (-0.25cm*1.25, 0.433cm*1.25);
    \end{tikzpicture}
    \caption{ A 2D plane with a gapped domain wall (red). For all the elementary disks that do not touch the red line,  \textbf{A1} continues to hold. However, for the elementary disks that touch the red line, \textbf{A1} may not necessarily hold. The axiom \textbf{A0} continues to hold everywhere.}
    \label{fig:gdw_setup}
\end{figure}
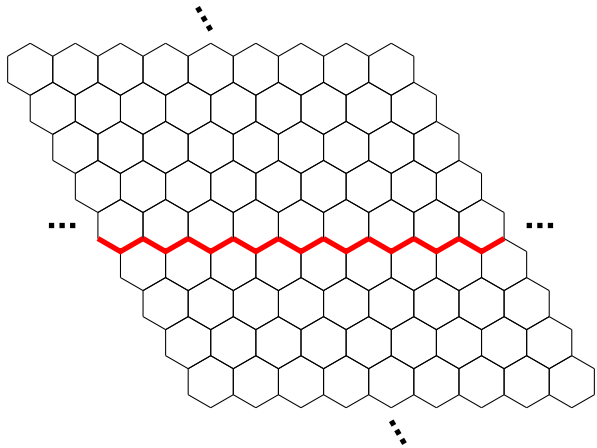

We now turn to a brief discussion on basic facts about quantum entropies [Section~\ref{sec:preliminary}]. The rest of the paper will focus on the proof of Theorem~\ref{theorem:main1}.

\subsection{Preliminary: Entropy inequalities}
\label{sec:preliminary}

Throughout this paper, density matrices are denoted by Greek letters, such as $\rho, \sigma, \tau, \ldots$ The subscripts refer to the Hilbert spaces the density matrices act on. For instance, $\rho_{ABC}$ is a density matrix acting on a Hilbert space $\mathcal{H}_{ABC}= \mathcal{H}_A\otimes \mathcal{H}_B\otimes \mathcal{H}_C$. The set of density matrices over subsystem $X$ is denoted as $\mathcal{D}_X$. The von Neumann entropy is defined as $S(\rho) = -\text{Tr}(\rho \log \rho)$. Entanglement entropy is defined as $S(A)_\rho = S(\rho_A)$. Conditional mutual information over a tripartite state $\rho_{ABC}$ is defined as 
\begin{equation}
    I(A:C|B)_{\rho} := S(AB)_{\rho} + S(BC)_{\rho}- S(B)_{\rho} - S(ABC)_{\rho}. \label{eq:ssa}
\end{equation}

We remark that the strong subadditivity of entropy implies the following inequalities, which we shall often use:
\begin{equation}
    \begin{aligned}
        &I(AA':CC'|B)_{\rho} \geq I(A:C|B)_{\rho}, \\
        &I(AA':C|B)_{\rho}  \geq I(A:C|A'B)_{\rho}.
    \end{aligned}
\end{equation}
We also remark that, for any tripartite quantum state $\rho_{ABC}$, the following inequality always holds.
\begin{equation}
    S(A|B)_{\rho} + S(A|C)_{\rho}\geq 0.\label{eq:wm}
\end{equation}
Eq.~\eqref{eq:wm} is known as the \emph{weak monotonicity}, a fact that we shall often use as well.

\section{Extensions of axioms}
\label{sec:extensions_of_axioms}

Recall our central assumption: that the reference state $\sigma$ satisfies \textbf{A0} and \textbf{A1} in the vicinity of every face $f\in \Lambda$ [Eq.~\eqref{eq:a0_starting_point},~\eqref{eq:a1_starting_point}]. One of the foundational facts in Ref.~\cite{shi2020fusion} is that the axioms \textbf{A0} and \textbf{A1} imply the same axioms at a larger scale. Brought to our current setup, this would imply, for instance, that 
\begin{equation}
    \left(S(C)  + S(C|B)  \right)_{\sigma} = 0 \quad \text{ for }\quad 
    \begin{tikzpicture}[scale=0.45, baseline={([yshift=-.55ex]current bounding box.center)}]
        \hexagonunitcolor{1}{1}{1}{red};
        \hexagonunitcolor{1}{-1}{0}{blue};
        \hexagonunitcolor{1}{-1}{1}{blue};
        \hexagonunitcolor{1}{-1}{2}{blue};
        \hexagonunitcolor{1}{0}{2}{blue};
        \hexagonunitcolor{1}{0}{0}{blue};
        \hexagonunitcolor{1}{1}{0}{blue};
        \hexagonunitcolor{1}{2}{0}{blue};
        \hexagonunitcolor{1}{0}{1}{red};
        \hexagonunitcolor{1}{2}{1}{blue};
        \hexagonunitcolor{1}{3}{1}{blue};
        \hexagonunitcolor{1}{1}{2}{blue};
        \hexagonunitcolor{1}{2}{2}{blue};
        \hexagonunitcolor{1}{3}{2}{blue};
        \node[] () at (-0.5, 0.866) {$C$};
        \node[] () at (-1.5, 0.866) {$B$};
    \end{tikzpicture} \label{eq:a0_extended_example}
\end{equation}
and 
\begin{equation}
    \left(S(C|B) + S(C|D) \right)_{\sigma} = 0 \quad \text{ for }\quad 
    \begin{tikzpicture}[scale=0.45, baseline={([yshift=-.55ex]current bounding box.center)}]
        \hexagonunitcolor{1}{1}{1}{red};
        \hexagonunitcolor{1}{-1}{0}{blue};
        \hexagonunitcolor{1}{-1}{1}{blue};
        \hexagonunitcolor{1}{-1}{2}{blue};
        \hexagonunitcolor{1}{0}{2}{green};
        \hexagonunitcolor{1}{0}{0}{blue};
        \hexagonunitcolor{1}{1}{0}{blue};
        \hexagonunitcolor{1}{2}{0}{blue};
        \hexagonunitcolor{1}{0}{1}{red};
        \hexagonunitcolor{1}{2}{1}{green};
        \hexagonunitcolor{1}{3}{1}{green};
        \hexagonunitcolor{1}{1}{2}{green};
        \hexagonunitcolor{1}{2}{2}{green};
        \hexagonunitcolor{1}{3}{2}{green};
        \node[] () at (-0.5, 0.866) {$C$};
        \node[] () at (-1.5, 0.866) {$B$};
        \node[] () at (1.5, 0.866) {$D$};
    \end{tikzpicture} \label{eq:a1_extended_example}
\end{equation}
An important similarity between Eq.~\eqref{eq:a0_starting_point} and Eq.~\eqref{eq:a0_extended_example} (and similarly, between Eq.~\eqref{eq:a1_starting_point} and Eq.~\eqref{eq:a1_extended_example}) is that a disk ($C$) is surrounded by the union of two disks ($B$ and $D$). Moreover, the union $BCD$ is again a disk. While the precise shapes of the subsystems labeled in $B$, $D$, $BC$, and $CD$ differ between Eq.~\eqref{eq:a1_starting_point} and Eq.~\eqref{eq:a1_extended_example}, topologically they are the same. Eq.~\eqref{eq:a0_extended_example} and \eqref{eq:a1_extended_example} are examples of what we mean by an extension of axioms \textbf{A0} and \textbf{A1}. While the argument for such extensions has been discussed in Ref.~\cite{shi2020fusion} already, we revisit this argument in this paper for the sake of concreteness. 

\subsection{Extending \textbf{A0}}
\label{sec:extending_a0}

Let us first consider the simpler case: proving extensions of \textbf{A0}. While the argument we present below is general, we aid the readers with a concrete example. Suppose, for instance, that our goal is to prove Eq.~\eqref{eq:a0_extended_example}. How would we proceed? The main idea can be schematically expressed as follows:
\begin{equation}
\begin{tikzpicture}[scale=0.45, baseline={([yshift=-.55ex]current bounding box.center)}]
        \hexagonunitcolor{1}{1}{1}{blue};
        \hexagonunitcolor{1}{-1}{0}{blue};
        \hexagonunitcolor{1}{-1}{1}{blue};
        \hexagonunitcolor{1}{-1}{2}{white};
        \hexagonunitcolor{1}{0}{2}{blue};
        \hexagonunitcolor{1}{0}{0}{blue};
        \hexagonunitcolor{1}{1}{0}{white};
        \hexagonunitcolor{1}{2}{0}{white};
        \hexagonunitcolor{1}{0}{1}{red};
        \hexagonunitcolor{1}{2}{1}{white};
        \hexagonunitcolor{1}{3}{1}{white};
        \hexagonunitcolor{1}{1}{2}{blue};
        \hexagonunitcolor{1}{2}{2}{white};
        \hexagonunitcolor{1}{3}{2}{white};
        \node[] () at (-0.5, 0.866) {$C$};
        \node[] () at (-1.5, 0.866) {$B$};
    \end{tikzpicture} 
\stackrel{\text{\textbf{A1}}}{\longrightarrow}
    \begin{tikzpicture}[scale=0.45, baseline={([yshift=-.55ex]current bounding box.center)}]
        \hexagonunitcolor{1}{1}{1}{red};
        \hexagonunitcolor{1}{-1}{0}{blue};
        \hexagonunitcolor{1}{-1}{1}{blue};
        \hexagonunitcolor{1}{-1}{2}{blue};
        \hexagonunitcolor{1}{0}{2}{blue};
        \hexagonunitcolor{1}{0}{0}{blue};
        \hexagonunitcolor{1}{1}{0}{blue};
        \hexagonunitcolor{1}{2}{0}{blue};
        \hexagonunitcolor{1}{0}{1}{red};
        \hexagonunitcolor{1}{2}{1}{blue};
        \hexagonunitcolor{1}{3}{1}{blue};
        \hexagonunitcolor{1}{1}{2}{blue};
        \hexagonunitcolor{1}{2}{2}{blue};
        \hexagonunitcolor{1}{3}{2}{blue};
        \node[] () at (-0.5, 0.866) {$C$};
        \node[] () at (-1.5, 0.866) {$B$};
    \end{tikzpicture}.
\end{equation}
This diagram means that we begin with $\left(S(C) + S(BC) - S(B) \right)_{\sigma}=0$ for the subsystems $B$ and $C$ on the left hand side and then gradually change the subsystems via $\textbf{A1}$ to arrive at $\left(S(C) + S(BC) - S(B) \right)_{\sigma}=0$ for the subsystems shown on the right hand side.

As a first step, we can first argue
\begin{equation}
\left( S(C) + S(C|B) \right)_{\sigma}=0 \quad \text{ for } \quad \begin{tikzpicture}[scale=0.45, baseline={([yshift=-.55ex]current bounding box.center)}]
        \hexagonunitcolor{1}{1}{1}{blue};
        \hexagonunitcolor{1}{-1}{0}{blue};
        \hexagonunitcolor{1}{-1}{1}{blue};
        \hexagonunitcolor{1}{-1}{2}{blue};
        \hexagonunitcolor{1}{0}{2}{blue};
        \hexagonunitcolor{1}{0}{0}{blue};
        \hexagonunitcolor{1}{1}{0}{blue};
        \hexagonunitcolor{1}{2}{0}{blue};
        \hexagonunitcolor{1}{0}{1}{red};
        \hexagonunitcolor{1}{2}{1}{blue};
        \hexagonunitcolor{1}{3}{1}{blue};
        \hexagonunitcolor{1}{1}{2}{blue};
        \hexagonunitcolor{1}{2}{2}{blue};
        \hexagonunitcolor{1}{3}{2}{blue};
        \node[] () at (-0.5, 0.866) {$C$};
        \node[] () at (-1.5, 0.866) {$B$};
    \end{tikzpicture}.
    \label{eq:a0_extension_first_argument}
\end{equation}
This is a simple consequence of SSA:
\begin{equation}
    \left(S(C) + S(C|B) \right)_{\sigma} \geq \left(S(C) + S(C|BB') \right)_{\sigma}
\end{equation}
for any subsystem $B'$. Now we can evaluate the difference of $\left(S(C) + S(C|B) \right)_{\sigma}$ appearing in Eq.~\eqref{eq:a0_extended_example} and Eq.~\eqref{eq:a0_extension_first_argument}, the result of which is
\begin{equation}
    \left(S(C|B) + S(C|D) \right)_{\sigma} \quad \text{ for } \quad 
    \begin{tikzpicture}[scale=0.45, baseline={([yshift=-.55ex]current bounding box.center)}]
        \hexagonunitcolor{1}{1}{1}{red};
        \hexagonunitcolor{1}{-1}{0}{green};
        \hexagonunitcolor{1}{-1}{1}{green};
        \hexagonunitcolor{1}{-1}{2}{green};
        \hexagonunitcolor{1}{0}{2}{green};
        \hexagonunitcolor{1}{0}{0}{green};
        \hexagonunitcolor{1}{1}{0}{green};
        \hexagonunitcolor{1}{2}{0}{green};
        \hexagonunitcolor{1}{0}{1}{blue};
        \hexagonunitcolor{1}{2}{1}{green};
        \hexagonunitcolor{1}{3}{1}{green};
        \hexagonunitcolor{1}{1}{2}{green};
        \hexagonunitcolor{1}{2}{2}{green};
        \hexagonunitcolor{1}{3}{2}{green};
        \node[] () at (0.5, 0.866) {$C$};
        \node[] () at (-0.5, 0.866) {$B$};
        \node[] () at (1.5, 0.866) {$D$};
    \end{tikzpicture}.
    \label{eq:a1_extension_second_step1}
\end{equation}
Conveniently, this quantity is zero. To see why, note that
\begin{equation}
    \left(S(C|B) + S(C|D) \right)_{\sigma}=0 \quad \text{ for } \quad 
    \begin{tikzpicture}[scale=0.45, baseline={([yshift=-.55ex]current bounding box.center)}]
        \hexagonunitcolor{1}{1}{1}{red};
        \hexagonunitcolor{1}{-1}{0}{white};
        \hexagonunitcolor{1}{-1}{1}{white};
        \hexagonunitcolor{1}{-1}{2}{white};
        \hexagonunitcolor{1}{0}{2}{white};
        \hexagonunitcolor{1}{0}{0}{green};
        \hexagonunitcolor{1}{1}{0}{green};
        \hexagonunitcolor{1}{2}{0}{white};
        \hexagonunitcolor{1}{0}{1}{blue};
        \hexagonunitcolor{1}{2}{1}{green};
        \hexagonunitcolor{1}{3}{1}{white};
        \hexagonunitcolor{1}{1}{2}{green};
        \hexagonunitcolor{1}{2}{2}{green};
        \hexagonunitcolor{1}{3}{2}{white};
        \node[] () at (0.5, 0.866) {$C$};
        \node[] () at (-0.5, 0.866) {$B$};
        \node[] () at (1.5, 0.866) {$D$};
    \end{tikzpicture}
    \label{eq:a1_extension_second_step2}
\end{equation}
by our assumption [Eq.~\eqref{eq:a0_starting_point}]. A simple consequence of SSA is that
\begin{equation}
    \left(S(C|B) + S(C|D) \right)_{\sigma} \geq \left(S(C|B) + S(C|DD') \right)_{\sigma}
\end{equation}
for any $D'$. Viewing the green faces in Eq.~\eqref{eq:a1_extension_second_step1} that are colored in white in Eq.~\eqref{eq:a1_extension_second_step2} as $D'$, we can see that $\left(S(C|B) + S(C|D) \right)_{\sigma} \leq 0$ for the choice $B, C,$ and $D$ shown in Eq.~\eqref{eq:a1_extension_second_step1}. At the same time, by the weak monotonicity, this quantity must be always zero. Thus the difference of $\left(S(C) + S(C|B) \right)_{\sigma}$ appearing in  Eq.~\eqref{eq:a0_extended_example} and Eq.~\eqref{eq:a0_extension_first_argument} are exactly equal to each other, and in fact, both zero. This completes the proof.

We can generalize this procedure into an algorithm that extends our axioms. Without loss of generality, suppose our task is to prove $\left(S(C) + S(C|B) \right)_{\sigma}=0$, where $C$ is a disk and $B$ is an annulus surrounding $B$, so that $BC$ together forms a disk. Pick any face $c\in C$ and let $b=  N(c)$. By assumption, $\left(S(c) + S(c|b) \right)_{\sigma}=0$. From SSA, it follows that $\left(S(c) + S(c|BC\setminus c) \right)_{\sigma}=0$. Now, the idea is to extend the disk $c$ into $C$, adding one face at a time. Specifically, consider a sequence of disks $\{C_i:i=1,\ldots n \}$ such that $c_1=c$, $C_n= C$ and $C_i\subset C_{i+1}$.\footnote{Such a sequence of disks can be constructed as follows. For each $C_i$, consider a face adjacent to $C_i$, which is still part of $C$, and then choose that to be $c_i$.}  Letting $c_{i+1} := C_{i+1}\setminus C_i$ be a face, note the following identity:
\begin{equation}
    \left(S(C_{i+1}) +S(C_{i+1}|BC\setminus C_{i+1}) - S(C_i) -S(C_i|BC\setminus C_i)\right)_{\sigma}= \left(S(c_{i+1}|C_i) + S(c_{i+1}| BC\setminus C_{i+1})\right)_{\sigma}. \label{eq:a1_extension_argument_general}
\end{equation}
Because $B\supset N(C)$ and $c_{i+1} \subset C$, $N(c_{i+1}) \subset BC$. Therefore, we can partition $N(c_{i+1})$ into $N(c_{i+1}) = (N(c_{i+1})_1 \cup N(c_{i+1})_2$, where
\begin{equation}
\begin{aligned}
    N(c_{i+1})_1 &:= N(c_{i+1}) \cap C_i,\\
    N(c_{i+1})_2 &:= N(c_{i+1}) \cap BC\setminus C_{i+1}.
\end{aligned}
\end{equation}
Note that, by SSA, the following inequalities hold:
\begin{equation}
    \begin{aligned}
        S(c_{i+1}|C_i) &\leq S(c_{i+1}|N(c_{i+1})_1 ), \\
        S(c_{i+1}|BC\setminus C_{i+1} )&\leq S(c_{i+1}|N(c_{i+1})_2),
    \end{aligned}
\end{equation}
which together with weak monotoncity implies that Eq.~\eqref{eq:a1_extension_argument_general} is zero. Iterating this argument from $i=1$ to $i=n-1$, we conclude that $\left(S(C) + S(C|B) \right)_{\sigma}=0$, proving our claim.

\subsection{Extending \textbf{A1}}
\label{sec:extending_a1}

Though different in details, the argument for extending \textbf{A1} is similar in spirit. Let us again work through a guiding example; below, we prove $\left(S(C|B) + S(C|D) \right)_{\sigma}=0$ for the choice of subsystems shown on right hand side of Eq.~\eqref{eq:a1_extension_example_setup}. This argument begins with $\left(S(C|B) + S(C|D) \right)_{\sigma}=0$ for the subsystems shown on the left hand side of Eq.~\eqref{eq:a1_extension_example_setup} (which can be viewed as the assumption \textbf{A1}), followed by applications of \textbf{A1} over other elementary disks.
\begin{equation}
    \begin{tikzpicture}[scale=0.45, baseline={([yshift=-.55ex]current bounding box.center)}]
        \hexagonunitcolor{1}{1}{1}{green};
        \hexagonunitcolor{1}{-1}{0}{blue};
        \hexagonunitcolor{1}{-1}{1}{blue};
        \hexagonunitcolor{1}{-1}{2}{white};
        \hexagonunitcolor{1}{0}{2}{green};
        \hexagonunitcolor{1}{0}{0}{blue};
        \hexagonunitcolor{1}{1}{0}{white};
        \hexagonunitcolor{1}{2}{0}{white};
        \hexagonunitcolor{1}{0}{1}{red};
        \hexagonunitcolor{1}{2}{1}{white};
        \hexagonunitcolor{1}{3}{1}{white};
        \hexagonunitcolor{1}{1}{2}{green};
        \hexagonunitcolor{1}{2}{2}{white};
        \hexagonunitcolor{1}{3}{2}{white};
        \hexagonunitcolor{1}{1}{3}{white};
        \hexagonunitcolor{1}{2}{3}{white};
        \node[] () at (-0.5, 0.866) {$C$};
        \node[] () at (-1.5, 0.866) {$B$};
        \node[] () at (0.5, 0.866) {$D$};
    \end{tikzpicture}
    \stackrel{\text{\textbf{A1}}}{\longrightarrow}
    \begin{tikzpicture}[scale=0.45, baseline={([yshift=-.55ex]current bounding box.center)}]
        \hexagonunitcolor{1}{1}{1}{red};
        \hexagonunitcolor{1}{-1}{0}{blue};
        \hexagonunitcolor{1}{-1}{1}{blue};
        \hexagonunitcolor{1}{-1}{2}{blue};
        \hexagonunitcolor{1}{0}{2}{green};
        \hexagonunitcolor{1}{0}{0}{blue};
        \hexagonunitcolor{1}{1}{0}{blue};
        \hexagonunitcolor{1}{2}{0}{blue};
        \hexagonunitcolor{1}{0}{1}{red};
        \hexagonunitcolor{1}{2}{1}{green};
        \hexagonunitcolor{1}{3}{1}{green};
        \hexagonunitcolor{1}{1}{2}{red};
        \hexagonunitcolor{1}{2}{2}{green};
        \hexagonunitcolor{1}{3}{2}{green};
        \hexagonunitcolor{1}{1}{3}{green};
        \hexagonunitcolor{1}{2}{3}{green};
        \node[] () at (-0.5, 0.866) {$C$};
        \node[] () at (-1.5, 0.866) {$B$};
        \node[] () at (1.5, 0.866) {$D$};
    \end{tikzpicture}.
    \label{eq:a1_extension_example_setup}
\end{equation}

Our argument is best explained in terms of the \emph{deformations of the boundaries}. Specifically, for any subsystem $A\subset \Lambda$, let us define $\partial A$ as the boundary of $A$, viewing $A$ as a subset of $\mathbb{R}^2$. These boundaries can be classified into the inner and the outer boundaries; the outer boundary is $\partial (BCD)$ and the inner boundaries are $\partial B\cap \partial C$, $\partial C \cap \partial D$, and $\partial B \cap \partial D$, drawn in green, blue, and red below:
\begin{equation}
    \begin{tikzpicture}[scale=0.45, baseline={([yshift=-.55ex]current bounding box.center)}]
        \hexagonunitcolor{1}{1}{1}{green};
        \hexagonunitcolor{1}{-1}{0}{blue};
        \hexagonunitcolor{1}{-1}{1}{blue};
        \hexagonunitcolor{1}{-1}{2}{white};
        \hexagonunitcolor{1}{0}{2}{green};
        \hexagonunitcolor{1}{0}{0}{blue};
        \hexagonunitcolor{1}{1}{0}{white};
        \hexagonunitcolor{1}{2}{0}{white};
        \hexagonunitcolor{1}{0}{1}{red};
        \hexagonunitcolor{1}{2}{1}{white};
        \hexagonunitcolor{1}{3}{1}{white};
        \hexagonunitcolor{1}{1}{2}{green};
        \hexagonunitcolor{1}{2}{2}{white};
        \hexagonunitcolor{1}{3}{2}{white};
        \hexagonunitcolor{1}{1}{3}{white};
        \hexagonunitcolor{1}{2}{3}{white};
        \node[] () at (-0.5, 0.866) {$C$};
        \node[] () at (-1.5, 0.866) {$B$};
        \node[] () at (0.5, 0.866) {$D$};
        \begin{scope}[yshift=0.04cm]
        \draw[line width=0.075cm, green] (-1, 0.866+0.25) --++ (0, -0.577) --++ (0.5, -0.25) --++ (0.5, 0.25); 
        \draw[line width=0.075cm, blue] (-1, 0.866+0.25) --++ (0.5, 0.25) --++ (0.5, -0.25) -- ++ (0, -0.577);
        \draw[line width=0.075cm, red] (-1, 0.866+0.25) --++ (-0.5, 0.25); 
        \draw[line width=0.075cm, red] (0, 0.866+0.25-0.577) --++ (0.5, -0.25); 
        \end{scope}
    \end{tikzpicture}
    \stackrel{\text{\textbf{A1}}}{\longrightarrow}
    \begin{tikzpicture}[scale=0.45, baseline={([yshift=-.55ex]current bounding box.center)}]
        \hexagonunitcolor{1}{1}{1}{red};
        \hexagonunitcolor{1}{-1}{0}{blue};
        \hexagonunitcolor{1}{-1}{1}{blue};
        \hexagonunitcolor{1}{-1}{2}{blue};
        \hexagonunitcolor{1}{0}{2}{green};
        \hexagonunitcolor{1}{0}{0}{blue};
        \hexagonunitcolor{1}{1}{0}{blue};
        \hexagonunitcolor{1}{2}{0}{blue};
        \hexagonunitcolor{1}{0}{1}{red};
        \hexagonunitcolor{1}{2}{1}{green};
        \hexagonunitcolor{1}{3}{1}{green};
        \hexagonunitcolor{1}{1}{2}{red};
        \hexagonunitcolor{1}{2}{2}{green};
        \hexagonunitcolor{1}{3}{2}{green};
        \hexagonunitcolor{1}{1}{3}{green};
        \hexagonunitcolor{1}{2}{3}{green};
        \node[] () at (-0.5, 0.866) {$C$};
        \node[] () at (-1.5, 0.866) {$B$};
        \node[] () at (1.5, 0.866) {$D$};
        \begin{scope}[yshift=0.04cm]
        \draw[line width=0.075cm, green] (-1, 0.866+0.25) --++ (0, -0.577) --++ (0.5, -0.25) --++ (0.5, 0.25) --++ (0.5, -0.25) -- ++ (0.5, 0.25); 
        \draw[line width=0.075cm, blue] (-1, 0.866+0.25) --++ (0.5, 0.25) --++ (0, 0.577) -- ++ (0.5, 0.25) -- ++ (0.5, -0.25) -- ++ (0, -0.577) -- ++ (0.5, -0.25) -- ++ (0, -0.577);
        \draw[line width=0.075cm, red] (-1, 0.866+0.25) --++ (-0.5, 0.25) --++ (0, 0.577); 
        \draw[line width=0.075cm, red] (1, 0.866+0.25-0.577) --++ (0.5, -0.25) --++ (0.5, 0.25) -- ++ (0.5, -0.25); 
        \end{scope}
    \end{tikzpicture}.
    \label{eq:a1_extension_example_setup_with_boundaries}
\end{equation}
We shall deform the shown boundaries, by deforming the subsystems ($BC$, $CD$, $B$, and $D$) by one face at a time whilst maintaining their topologies. The key point here is that the value of $\left(S(C|B) + S(C|D) \right)_{\sigma}$ remains invariant under two types of moves that deform the boundaries. By combining these two moves together, one can conclude $\left(S(C|B) + S(C|D) \right)_{\sigma}$ on the right hand side of Eq.~\eqref{eq:a1_extension_example_setup_with_boundaries} is zero.

The first move is to deform the outer boundaries, which entails enlarging/reducing either $B$ or $D$. This move keeps the green and the blue boundary unchanged, but it may deform the red boundaries. Since the quantity at hand $\left(S(C|B) + S(C|D) \right)_{\sigma}$ is symmetric under the exchange of $B$ and $D$, an argument that works for one of them works for the other as well. We thus focus on enlarging/reducing $D$. Without loss of generality, let us consider the difference between $\left(S(C|B) + S(C|D) \right)_{\sigma}$ and $\left(S(C|B) + S(C|Dd) \right)_{\sigma}$, where $d\in \Lambda \setminus (BCD)$ is a face. The difference between the two quantities is $I(d:C|D)_{\sigma}$.  
By SSA,
\begin{equation}
    I(d:C|D)_{\sigma} \leq I(d:CD\setminus N(d)|D\cap N(d))_{\sigma},
    \label{}
\end{equation}
where we used the fact that $C$ has a trivial intersection with $N(d)$.\footnote{Recall that $C$ is chosen in such a way that it is surrounded by $B$ and $D$. Therefore, $C$ cannot be in the neighborhood of $d$, which is in $\Lambda \setminus (BCD)$.} Further, 
\begin{equation}
    I(d:CD\setminus N(d)|D\cap N(d))_{\sigma} \leq \left(S(d|D\cap N(d)) + S(d|N(d)\setminus D) \right)_{\sigma}. \label{eq:22}
\end{equation}
By \textbf{A1}, the right hand side Eq.~\eqref{eq:22} is zero if $D\cap N(d)$ and $N(d)\setminus D$ are both disks. 

The condition that these subsystems ought to be disks can be ensured for any face $d$ as long as both $D$ and $Dd$ are disks. Suppose otherwise. That is, there is a face $d$ such that both $D$ and $Dd$ are disks but $N(d)\cap D$ is not a disk. Thus $N(d)\cap D$ must be a union of at least two disjoint disks. Moreover, there is a disk in $N(d)\setminus D$ --- surrounded by the aforementioned two disks, $d$, and $D$ --- whose neighborhood is contained in $Dd$. Any loop that surrounds the aforementioned disk in $N(d)\setminus D$ is not contractible to a point, so $Dd$ is not simply connected, which is a contradiction. 

Therefore, provided that $D$ and $Dd$ are disks, the right hand side of Eq.~\eqref{eq:22} is at most $0$. Since conditional mutual information is non-negative,  $I(d:C|D)_{\sigma}=0$. Thus 
\begin{equation}
    \left(S(C|B) + S(C|D) \right)_{\sigma} = \left(S(C|B) + S(C|Dd) \right)_{\sigma}.
\end{equation}
The same argument can be employed to argue the invariance under deformation of $B$ which retains its disk-like topology. Thus we conclude that $\left(S(C|B) + S(C|D) \right)_{\sigma}=0$ for the subsystems shown below:
\begin{equation}
    \begin{tikzpicture}[scale=0.45, baseline={([yshift=-.55ex]current bounding box.center)}]
        \hexagonunitcolor{1}{1}{1}{green};
        \hexagonunitcolor{1}{-1}{0}{blue};
        \hexagonunitcolor{1}{-1}{1}{blue};
        \hexagonunitcolor{1}{-1}{2}{blue};
        \hexagonunitcolor{1}{0}{2}{green};
        \hexagonunitcolor{1}{0}{0}{blue};
        \hexagonunitcolor{1}{1}{0}{blue};
        \hexagonunitcolor{1}{2}{0}{blue};
        \hexagonunitcolor{1}{0}{1}{red};
        \hexagonunitcolor{1}{2}{1}{green};
        \hexagonunitcolor{1}{3}{1}{green};
        \hexagonunitcolor{1}{1}{2}{green};
        \hexagonunitcolor{1}{2}{2}{green};
        \hexagonunitcolor{1}{3}{2}{green};
        \hexagonunitcolor{1}{1}{3}{green};
        \hexagonunitcolor{1}{2}{3}{green};
        \node[] () at (-0.5, 0.866) {$C$};
        \node[] () at (-1.5, 0.866) {$B$};
        \node[] () at (1.5, 0.866) {$D$};
        \begin{scope}[yshift=0.04cm]
        \draw[line width=0.075cm, green] (-1, 0.866+0.25) --++ (0, -0.577) --++ (0.5, -0.25) --++ (0.5, 0.25); 
        \draw[line width=0.075cm, blue] (-1, 0.866+0.25) --++ (0.5, 0.25) --++ (0.5, -0.25) -- ++ (0, -0.577);
        \draw[line width=0.075cm, red] (-1, 0.866+0.25) --++ (-0.5, 0.25) --++ (0, 0.577); 
        \draw[line width=0.075cm, red] (0, 0.866+0.25-0.577) --++ (0.5, -0.25) --++ (0.5, 0.25) -- ++ (0.5, -0.25) -- ++ (0.5, 0.25) -- ++ (0.5, -0.25); 
        \end{scope}
    \end{tikzpicture}.
    \label{eq:a1_extension_example_setup_with_boundaries_outer}
\end{equation}
Comparing this diagram to the right hand side of Eq.~\eqref{eq:a1_extension_example_setup_with_boundaries}, it is clear now the main task is to deform the inner boundaries. 

The argument for deforming the inner boundaries is exactly the same. Let us introduce a purifying system for $BCD$, denoted as $E$. Note the following identity:
\begin{equation}
    \left(S(C|B) + S(C|D) \right)_{\sigma} = \left(S(E|B) + S(E|D) \right)_{\sigma},
\end{equation}
which follows easily from the fact that the state over $BCDE$ is a pure state. The role played by $C$ in our previous argument can now be played by $E$. The new inner boundaries associated with  $B, E$ and $D$ (as opposed to $B, C,$ and $D$) is shown below:
\begin{equation}
    \begin{tikzpicture}[scale=0.45, baseline={([yshift=-.55ex]current bounding box.center)}]
        \hexagonunitcolor{1}{1}{1}{green};
        \hexagonunitcolor{1}{-1}{0}{blue};
        \hexagonunitcolor{1}{-1}{1}{blue};
        \hexagonunitcolor{1}{-1}{2}{blue};
        \hexagonunitcolor{1}{0}{2}{green};
        \hexagonunitcolor{1}{0}{0}{blue};
        \hexagonunitcolor{1}{1}{0}{blue};
        \hexagonunitcolor{1}{2}{0}{blue};
        \hexagonunitcolor{1}{2}{1}{green};
        \hexagonunitcolor{1}{3}{1}{green};
        \hexagonunitcolor{1}{1}{2}{green};
        \hexagonunitcolor{1}{2}{2}{green};
        \hexagonunitcolor{1}{3}{2}{green};
        \hexagonunitcolor{1}{1}{3}{green};
        \hexagonunitcolor{1}{2}{3}{green};
        \node[] () at (1, -0.866-0.433) {$E$ (Purifying system)};
        \node[] () at (-1.5, 0.866) {$B$};
        \node[] () at (1.5, 0.866) {$D$};
        \begin{scope}[yshift=0.04cm]
        \draw[line width=0.075cm, green] (-1.5, 0.866+0.5+0.577) --++ (-0.5, 0.25) --++ (-0.5, -0.25) --++ (0, -0.577) --++ (0.5, -0.25) --++ (0, -0.577) --++ (0.5, -0.25) --++ (0, -0.577) --++ (0.5, -0.25) --++ (0.5, 0.25)--++ (0.5, -0.25) --++ (0.5, 0.25)--++ (0.5, -0.25) --++ (0.5, 0.25)--++ (0.5, -0.25) --++ (0.5, 0.25) --++ (0, 0.577); 
        \draw[line width=0.075cm, blue] (-1.5, 0.866+0.5+0.577) --++ (0.5, 0.25) --++ (0, 0.577) --++ (0.5, 0.25) --++ (0.5, -0.25) --++ (0.5, 0.25) --++ (0.5, -0.25) --++ (0, -0.577) --++ (0.5, -0.25) --++ (0.5, 0.25) --++ (0.5, -0.25) --++ (0, -0.577) --++ (0.5, -0.25) --++ (0, -0.577) --++ (-0.5, -0.25);
        \draw[line width=0.075cm, red] (-1, 0.866+0.25) --++ (-0.5, 0.25) --++ (0, 0.577); 
        \draw[line width=0.075cm, red] (0, 0.866+0.25-0.577) --++ (0.5, -0.25) --++ (0.5, 0.25) -- ++ (0.5, -0.25) -- ++ (0.5, 0.25) -- ++ (0.5, -0.25); 
        \end{scope}
    \end{tikzpicture},
    \label{eq:a1_extension_example_setup_with_boundaries_inside_out}
\end{equation}
where the purifying system $E$ is not specified explicitly\footnote{Note that $E$ is not necessarily $\Lambda \setminus (BCD)$, but rather an abstract system that purifies $\sigma_{BCD}$.}. Note that the outer boundaries with respect to $B, E,$ and $D$ are precisely the green and the blue boundaries with respect to $B, C,$ and $D$. The argument for deforming these new outer boundaries is exactly the same as before, leading to
\begin{equation}
    \left(S(E|B) + S(E|D) \right)_{\sigma}=0 \quad \text{ for }\quad
    \begin{tikzpicture}[scale=0.45, baseline={([yshift=-.55ex]current bounding box.center)}]
        \hexagonunitcolor{1}{-1}{0}{blue};
        \hexagonunitcolor{1}{-1}{1}{blue};
        \hexagonunitcolor{1}{-1}{2}{blue};
        \hexagonunitcolor{1}{0}{2}{green};
        \hexagonunitcolor{1}{0}{0}{blue};
        \hexagonunitcolor{1}{1}{0}{blue};
        \hexagonunitcolor{1}{2}{0}{blue};
        \hexagonunitcolor{1}{2}{1}{green};
        \hexagonunitcolor{1}{3}{1}{green};
        \hexagonunitcolor{1}{2}{2}{green};
        \hexagonunitcolor{1}{3}{2}{green};
        \hexagonunitcolor{1}{1}{3}{green};
        \hexagonunitcolor{1}{2}{3}{green};
        \node[] () at (1, -0.866-0.433) {$E$ (Purifying system)};
        \node[] () at (-1.5, 0.866) {$B$};
        \node[] () at (1.5, 0.866) {$D$};
        \begin{scope}[yshift=0.04cm]
        \draw[line width=0.075cm, green] (-1.5, 0.866+0.5+0.577) --++ (-0.5, 0.25) --++ (-0.5, -0.25) --++ (0, -0.577) --++ (0.5, -0.25) --++ (0, -0.577) --++ (0.5, -0.25) --++ (0, -0.577) --++ (0.5, -0.25) --++ (0.5, 0.25)--++ (0.5, -0.25) --++ (0.5, 0.25)--++ (0.5, -0.25) --++ (0.5, 0.25)--++ (0.5, -0.25) --++ (0.5, 0.25) --++ (0, 0.577); 
        \draw[line width=0.075cm, blue] (-1.5, 0.866+0.5+0.577) --++ (0.5, 0.25) --++ (0, 0.577) --++ (0.5, 0.25) --++ (0.5, -0.25) --++ (0.5, 0.25) --++ (0.5, -0.25) --++ (0, -0.577) --++ (0.5, -0.25) --++ (0.5, 0.25) --++ (0.5, -0.25) --++ (0, -0.577) --++ (0.5, -0.25) --++ (0, -0.577) --++ (-0.5, -0.25);
        \draw[line width=0.075cm, red] (-1, 0.866+0.25) --++ (-0.5, 0.25) --++ (0, 0.577); 
        \draw[line width=0.075cm, red] (1, 0.866+0.25-0.577) --++ (0.5, -0.25) --++ (0.5, 0.25) -- ++ (0.5, -0.25); 
        \end{scope}
    \end{tikzpicture},
\end{equation}
which is equivalent to $\left(S(C|B) + S(C|D) \right)_{\sigma}=0$ for the subsystems shown on the right hand side of Eq.~\eqref{eq:a1_extension_example_setup_with_boundaries}.

The argument we employed in our example can be generalized to an algorithm. Without loss of generality, suppose our task is to show $\left(S(C|B) + S(C|D) \right)_{\sigma}$ for $B, C,$ and $D$ such that $BC, CD, B,$ and $D$ are disks and $BD$ is an annulus. First choose any $c\in C$. Now we choose $b, d \subset N(c)$ based on the red boundaries defined with respect to $B, C,$ and $D$. If a red boundary has a nontrivial overlap with $N(c)$, choose $b$ and $d$ such that the red boundaries defined with respect to $b, c,$ and $d$ overlap with the red boundaries defined with respect to $B, C,$ and $D$. Then, we can grow the red boundaries defined with respect to $b, c,$ and $d$ to the extent that they include the red boundaries defined with respect to $B, C,$ and $D$; this can be done by extending $b$ and $d$ by the blue and green faces adjacent to a unit interval of the red boundary. We can then trim the red boundary by introducing the purifying system and repeating the same argument, but this time reducing the red boundary. This process ensures that we have a choice of $B', C', D'$ such that their red boundaries match the red boundaries of $B, C,$ and $D$. Now, the outer and the blue/green boundaries can be deformed using the following greedy algorithm. For $D'$, at each time step, look at all the cells in $N(D')$ and if $N(D')\cap D$ is non-empty and if there is an element whose addition does not change the topology, add that. The same applies to $B'$. Then remove the unnecessary faces. 

While our results do not rely on showing
this algorithm works unconditionally for any subsystems topologically equivalent to the ones appearing in \textbf{A0} and \textbf{A1}, we expect that this is the case.\footnote{In fact, we suspect such an argument works even for tessellations of hyperbolic spaces.} Regardless, for the purpose of our paper, the description of the algorithm we provide shall be sufficient. This is because we simply need to extend \textbf{A0} and \textbf{A1} to a particular (finite) set of subsystems, for which one can verify that our algorithm works. These subsystems shall be introduced in Section~\ref{sec:parent_hamiltonian_final} and their construction shall be modified from the discussion in Section~\ref{sec:ltqo}, which we now turn our attention to.

\section{LTQO and local uniqueness}
\label{sec:ltqo}

In this section, we introduce the notion of local topological quantum order (LTQO), or having a locally unique ground state. Then we introduce a class of parent Hamiltonians for the reference state that satisfies this condition.

To produce a useful parent Hamiltonian, we want to establish some sense in which the ground state is unique.  While models such as the toric code  have degeneracies depending on the  topology of the manifold, they still satisfy a certain ``local uniqueness'' property.  In particular, for all topological disks $A$, every zero-energy state on $A$ looks identical on the interior of $A$.  In fact this is a reformulation of the condition introduced as  ``local topological quantum order'' (LTQO) by ~\cite{Michalakis2013}.  They show that LTQO implies the gap is stable to local perturbations.

 For simplicity, we define LTQO as a condition placed upon local, frustration-free Hamiltonians.  Thus the condition will be a special case of (\cite{Michalakis2013}, Definition 4). 

\begin{definition}[Local topological quantum order] 
\label{definition:ltqo}
    Consider a local frustration-free Hamiltonian, $H = \sum_X (1-Q_X)$, summed over some $O(1)$-size disks $X$, with projections $Q_X$ local to $X$. Let $Q_A$ denote the projection onto  $\{|\psi\rangle_A \in \mathcal{H}_A \,:\, (1-Q_X) |\psi\rangle_A= 0\; \forall X \subset A \}$, the space of states on $A$ with zero energy. Let $A$ be a ball of radius $r$ and $A(\ell)$ be the set of sites that are distance at most $\ell$ away from $A$. Let 
    \begin{equation}
        c_{\ell}(O_A) := \frac{\text{Tr}(Q_{A(\ell)} O_A)}{\text{Tr}(Q_{A(\ell)})}
    \end{equation}
    for any operator $O_A$ acting on $\mathcal{H}_A$. Then we say $H$ satisfies the condition of \emph{local topological quantum order} if there exists $\ell = O(1)$ independent of $r$, such that for all balls $A$ and operators $O_A$,
    \begin{equation}
      Q_{A(\ell)} O_A Q_{A(\ell)} =  c_{\ell}(O_A) Q_{A(\ell)}. \label{eq:ltqo_equation}
    \end{equation}
\end{definition}
\noindent
The original definition of LTQO allows Eq.~\eqref{eq:ltqo_equation} to be satisfied approximately. We shall be content with this simplified version because we can show Eq.~\eqref{eq:ltqo_equation} \emph{exactly} for $H$. In particular, if our version [Definition~\ref{definition:ltqo}] holds, the original version of LTQO holds as well.

Perhaps a more physical formulation is the following. 
\begin{definition}[Locally unique ground state]
\label{definition:local-uniqueness}
Consider a Hamiltonian $H$ as in Definition \ref{definition:ltqo}. Then we say $H$ has a locally unique ground state if, for some $\ell=O(1)$ independent of $r$, for all balls $A$ of any radius $r$, for all states $\rho_{A(\ell)}$ on enlarged region $A(\ell)$ that have zero energy on $A(\ell)$, the state
\begin{align} \label{eq:local-uniqueness}
    \tau_A := \Tr_{A(\ell) \backslash A} \rho_{A(\ell)} 
\end{align}
is independent of the state $\rho_{A(\ell)}$.  (We say that a state $\rho_R$ on region $R$ has zero energy if it is in the ground space of all the terms in $H$ supported on $R$.) 
\end{definition}
These formulations are equivalent.
\begin{proposition}[LTQO as local uniqueness] 
\label{prop:uniqueness}
For local, frustration-free Hamiltonians, the LTQO condition stated in Definition \ref{definition:ltqo} is equivalent to the local uniqueness condition of Definition
\ref{definition:local-uniqueness}.
\end{proposition}
The equivalence is known and straightforward.  For completeness we include a proof of the less direct implication, that local uniqueness implies LTQO.
\begin{proof}
We assume $H$ satisfies \eqref{eq:local-uniqueness} and show \eqref{eq:ltqo_equation}.
By \eqref{eq:local-uniqueness}, for any $|\Psi\rangle, |\Psi'\rangle \in \ker (I-Q_{A(\ell)})$ for a sufficiently large but constant $\ell$, 
\begin{equation}
    \langle \Psi| O_A |\Psi\rangle
    =
    \langle \Psi' | O_A |\Psi'\rangle.
\end{equation}
Choosing $|\Psi\rangle = \frac{1}{\sqrt{2}}(|\Psi_1\rangle + |\Psi_2\rangle)$ and $|\Psi'\rangle = \frac{1}{\sqrt{2}}(|\Psi_1\rangle - |\Psi_2\rangle)$ for an orthogonal pair of states $|\Psi_1\rangle$ and $|\Psi_2\rangle$, we can conclude that
\begin{equation}
    \langle \Psi_1|O_A |\Psi_2\rangle + \langle \Psi_2| O_A |\Psi_1\rangle=0.
\end{equation}
Choosing $|\Psi\rangle = \frac{1}{\sqrt{2}}(|\Psi_1\rangle + i|\Psi_2\rangle)$ and $|\Psi'\rangle = \frac{1}{\sqrt{2}}(|\Psi_1\rangle - i|\Psi_2\rangle)$, we can conclude that
\begin{equation}
    \langle \Psi_1|O_A |\Psi_2\rangle - \langle \Psi_2| O_A |\Psi_1\rangle=0,
\end{equation}
implying $\langle \Psi_1| O_A |\Psi_2\rangle = \langle \Psi_2|O_A |\Psi_1\rangle=0$. Therefore, the off-diagonal elements of $O_A$ in $\ker (I- Q_{A(\ell)})$ are identically zero. This immediately implies Eq.~\eqref{eq:ltqo_equation}.
\end{proof}

Now we turn to constructing a parent Hamiltonian for the reference state.  We consider a Hamiltonian of the following form,
\begin{equation}
    H = \sum_{X\in \mathcal{S}} (I-P_X), \label{eq:hamiltonian_global}
\end{equation}
where the elements of $\mathcal{S}$ are disks of constant size and
\begin{align}
P_X=\text{proj}_{\text{supp}(\rho_X)}
\end{align} 
is the projector onto the support of $\sigma_X$. 
The main result of this section will be
\begin{proposition}
\label{proposition:ltqo}[Parent Hamiltonian satisfies LTQO]
    Assume the Hamiltonian defined in Eq.~\eqref{eq:hamiltonian_global} uses a collection of regions $\mathcal{S}$ such that, for every elementary disk $D$, there exists a disk $X\in \mathcal{S}$ such that $D\cup N(D)\subset X$.  Then $H$ satisfies LTQO [Definition~\ref{definition:ltqo}].
\end{proposition}
One can think of the condition above as requiring that the regions $X$ cover the plane with sufficient buffer. While this condition is sufficient for LTQO, later we will make further specifications of $\mathcal{S}$.  With these further properties, the resulting Hamiltonian will be frustration-free and commuting.

Note ``support'' is a word used in two different ways: (i) the support of an operator is the subsystem $X$ that the operator acts on, and (ii) the support of a density matrix $\sigma_X$ is the orthogonal complement of the kernel, as a linear subspace of the Hilbert space on $X$. Since $\sigma$ lives in the ground state subspace by definition and $H\geq 0$, the ground state energy is zero. 

To prove Proposition \ref{proposition:ltqo}, we will show the local uniqueness property of Definition \ref{definition:local-uniqueness}, which is equivalent to LTQO [Definition \ref{definition:ltqo}] via Proposition \ref{prop:uniqueness}.  We begin with the following simple fact~\cite{shi2020fusion}.
\begin{lemma}
\label{lemma:a0_consequence}
    Let $\rho_{BC}$ be a density matrix such that $\left( S(C) + S(B|C)\right)_{\rho}=0$. For any density matrix $\tau_{BC}$ in the support of $\rho_{BC}$, $\tau_C = \rho_C$.
\end{lemma}
\begin{proof}
Note that $I(A:C)_{\psi_{\rho}}=0$, where $A$ is the purifying system. Here $\psi_{\rho} := {|\psi_{\rho}\rangle}_{ABC}{\langle \psi_{\rho}|}$ for some state vector $|\psi_{\rho}\rangle_{ABC}$. Thus $\rho_{AC} = \rho_A\otimes \rho_C$. By Uhlmann's theorem~\cite{Uhlmann1976}, this implies that there exists an isometry $V_{B_LB_R\to B}: \mathcal{H}_{B_L}\otimes \mathcal{H}_{B_R} \to \mathcal{H}_B $, where $B_L$ and $B_R$ are auxiliary systems purifying $\rho_A$ and $\rho_C$,
\begin{equation}
    |\psi_{\rho}\rangle_{ABC} = V_{B_LB_R \to B} |\psi_1\rangle_{AB_L} \otimes |\psi_2\rangle_{B_RC}.
\end{equation}
Thus the density matrices in the support of $\rho_{BC}$ must be of the following form:
\begin{equation}
    V_{ B_LB_R \to B}\left(\lambda_{B_L} \otimes |\psi_2\rangle_{B_RC}\langle\psi_2| \right) V_{B_LB_R \to B}^{\dagger}, \label{eq:tau_canonical_form}
\end{equation}
where $\lambda_{B_L}$ is a density matrix acting on the support of $\text{Tr}_A( |\psi_1\rangle_{AB_L}\langle\psi_1|)$. Upon tracing out $B$ in Eq.~\eqref{eq:tau_canonical_form}, we obtain a density matrix $\text{Tr}_{B_R}(|\psi_2\rangle_{B_RC}\langle\psi_2|)$, which is precisely $\rho_C$.
\end{proof}

Using Lemma~\ref{lemma:a0_consequence}, we can conclude that \emph{any} element in $\ker H$, restricted to any elementary disk, ought to be exactly the marginal of $\sigma$ defined on the same elementary disk. To see why, recall that we assumed that for any elementary disk $D$, $D\cup N(D)$ is included in some $X\subset \mathcal{S}$. Therefore, any element in $\ker H$, restricted to $D\cup N(D)$, must be supported on the support of $\sigma_{D\cup N(D)}$. Setting $C=D$ and $B=N(D)$, we recall that [Section~\ref{sec:extending_a0}]
\begin{equation}
    \left(S(C) + S(B|C) \right)_{\sigma}=0.
\end{equation}
By Lemma~\ref{lemma:a0_consequence}, $\rho_C = \sigma_C$, where $C=D$ is an elementary disk. Thus for any elementary disk $D$, for any density matrix $\rho$ supported in $\ker H$,
\begin{equation}
   \rho_D= \sigma_D.
\end{equation}

This argument can be easily generalized to a Hamiltonian restricted to any finite disk. Consider a disk, denoted as $D\subset \Lambda$. Define
\begin{equation}
    H_{D} := \sum_{\substack{X\in \mathcal{S},\\X \subset D}} (I-P_X).
\end{equation}
We claim that any element in $\ker H_D$, restricted to any elementary disk whose distance from $\partial D$ is larger than some constant of order unity, is equal to the reduced density matrix of $\sigma$ on the same elementary disk. Recall that $X\subset \mathcal{S}$ is a disk of constant size. Let $r$ be the maximal radius of $X$. For any elementary disk $u$ that has a distance (to $\partial D$) larger than $2r+1$, $u\cup N(u)$ cannot be included in $X\subset \mathcal{S}$ that lies outside of $D$. On the other hand, by our construction, $u\cup N(u)$ ought to be a subset of some $X\subset \mathcal{S}$. Therefore, for any elementary disk $u$ whose distance from $\partial D$ is larger than $2r+1$, there ought to be a $X\subset \mathcal{S}$ such that $u\cup N(u)\subset X$. Again using Lemma~\ref{lemma:a0_consequence}, the reduced density matrix of any element of $\ker H_D$ over $u$ is exactly $\sigma_D$.

Therefore, for any disk $D$, for any $D'\subset D$ separated from $\partial D$ by at least some positive distance, the following property holds. Without loss of generality, consider an elementary disk associated with a face in $D'$. The reduced density matrix over this elementary disk (obtained from an element in $\ker H_D$) is indistinguishable from the reduced density matrix of $\sigma$ over the same elementary disk. 

In fact, more can be said. Let $\rho_{D'}$ be a reduced density matrix of an element of $\ker H_D$ over $D'$. We claim
\begin{equation}
    \rho_{D'} = \sigma_{D'}. \label{eq:ltqo_identity}
\end{equation}
The proof is based on induction. Let $D_i\subset D'$ be a disk and suppose $\rho_{D_i} = \sigma_{D_i}$. Consider a face $f\in D'$, where $D_i f$ is a disk. From SSA, it follows that:
\begin{equation}
    \begin{aligned}
        I(f:D_i \setminus N(f)|D_i\cap N(f))_{\sigma} &\leq  \left(S(f|D_i\cap N(f))  + S(f|N(f)\setminus D_i)\right)_{\sigma}, \\
        I(f:D_i \setminus N(f)|D_i\cap N(f))_{\rho} &\leq  \left(S(f|D_i\cap N(f))  + S(f|N(f)\setminus D_i)\right)_{\rho}.
    \end{aligned}
\end{equation}
Both $\rho$ and $\sigma$ are indistinguishable over $D_i$ and $f\cup (N(f)\cap D_i)$ by assumption. Since two locally indistinguishable Markov chains are globally indistinguishable [Lemma~\ref{lemma:markov_local_to_global}], $\rho_{f D_i} = \sigma_{fD_i}$. This argument can be repeated until the enlarged disk becomes $D'.$ This completes the proof of  Proposition \ref{proposition:ltqo}.  

We conclude with the lemma used above,
\begin{lemma}
    Let $\rho_{ABC}$ and $\sigma_{ABC}$ be quantum Markov chains. If $\rho_{AB} = \sigma_{AB}$ and $\rho_{BC} = \sigma_{BC}$, $\rho_{ABC} = \sigma_{ABC}$.
    \label{lemma:markov_local_to_global}
\end{lemma}
\begin{proof}
The following proof is based on Ref.~\cite{Kim2014informational}. Let $\tau_{ABC} =\frac{1}{2}(\rho_{ABC} + \sigma_{ABC})$. 
\begin{equation}
    S(\tau_{ABC}) - \frac{1}{2} (S(\rho_{ABC}) + S(\sigma_{ABC})) \leq \frac{1}{2}(I(A:C|B)_{\rho} + I(A:C|B)_{\sigma}).
\end{equation}
The right hand side is zero by our assumption. By the strict concavity of the von Neumann entropy, the left hand side must be zero. Moreover, this is possible if and only if $\rho_{ABC} = \sigma_{ABC}$.
\end{proof}

\section{Commuting projectors from quantum Markov chain}
\label{sec:commuting_projectors}

The Hamiltonian defined in Eq.~\eqref{eq:hamiltonian_global} satisfies the LTQO condition [Proposition~\ref{proposition:ltqo}]. However, in order to use the stability of the gap proved in Ref.~\cite{Michalakis2013}, one needs an additional ingredient: local gap condition. What this means is that $H_X$, for any disk $X$, has a spectral gap bounded from below between the ground state subspace of $H_X$ and the rest of the spectrum. The purpose of this Section is to develop techniques that can establish such a result.

The main result of this Section is Proposition~\ref{proposition:commutation_markov}. This proposition states that, if $\rho_{ABC}$ is a quantum Markov chain, the projectors onto the support of $\rho_{AB}$ and $\rho_{BC}$ commute.

\begin{theorem}
\cite{Hayden2004} A tripartite quantum state $\rho_{ABC}$, defined on a finite-dimensional Hilbert space $\mathcal{H}_A\otimes \mathcal{H}_B \otimes \mathcal{H}_C$, satisfies $I(A:C|B)_{\rho}=0$ if and only if there is a decomposition of $B$ as 
\begin{equation}
    \mathcal{H}_B = \bigoplus_{j} \mathcal{H}_{b_j^L} \otimes \mathcal{H}_{b_j^R}
\end{equation}
such that 
\begin{equation}
    \rho_{ABC} = \bigoplus_j q_j \rho_{Ab_j^L} \otimes \rho_{b_j^R C} \label{eq:quantum_markov_chain_decomposition}
\end{equation}
with states $\rho_{Ab_j^L}$ in $\mathcal{H}_{Ab_j^L}$ and $\rho_{b_j^R C}$ in $\mathcal{H}_{b_j^RC}$, and a probability distribution $\{q_j \}$.
\end{theorem}

\begin{lemma}
\label{lemma:lemma0}
Consider a bipartite Hilbert space $\mathcal{H}_{AB}=\mathcal{H}_A \otimes \mathcal{H}_B$ and let $\rho_{AB}$ be a state that acts on it. Let $P_{AB}$ and $P_B$ be the projector onto the support of $\rho_{AB}$ and $\rho_B$, respectively. 
\begin{equation}
    [P_B, P_{AB}]=0.
\end{equation}
\end{lemma}
\begin{proof}
Consider the images of $I\otimes P_B$ and $P_{AB}$. The latter is a subspace of the former. Within the former subspace, $I\otimes P_B$ commutes with every element, and in particular, with $P_{AB}$. In the orthogonal complement, both operators annihilate every state, and as such, trivially commute.
\end{proof}

\begin{proposition}
\label{proposition:commutation_markov}
Let $\rho_{ABC}$ be a state defined on a finite dimensional Hilbert space $\mathcal{H}_A\otimes \mathcal{H}_B \otimes \mathcal{H}_C$ that satisfies $I(A:C|B)_{\rho}=0$. Let $P_{AB}$ and $P_{BC}$ be the projector onto the support of $\rho_{AB}$ and $\rho_{BC}$, respectively. Then 
\begin{equation}
    [P_{AB}, P_{BC}]=0.
\end{equation}
\end{proposition}
\begin{proof}
Consider the decomposition in Eq.~\eqref{eq:quantum_markov_chain_decomposition}. Let $J = \{ j: q_j \neq 0\}$. Define $P_{Ab_j^L}$ and $P_{b_j^RC}$ as projectors onto the support of $\rho_{Ab_j^L}$ and $\rho_{b_j^RC}$, respectively. Similarly, define $P_{b_j^L}$ and $P_{b_j^R}$ to be the projectors onto the support of $\rho_{b_j^L}$ and $\rho_{b_j^R}$, respectively. These projectors are related to $P_{AB}$ and $P_{BC}$ in the following way:
\begin{equation}
    \begin{aligned}
    P_{AB} = \sum_{j\in J} P_{Ab_j^L} \otimes P_{b_j^R}, \\
    P_{BC} = \sum_{j\in J} P_{b_j^L} \otimes P_{b_j^RC}.
    \end{aligned}
    \label{eq:projectors_decomposition}
\end{equation}
Since that $[P_{b_j^L}, P_{Ab_j^L}]=0$ and $[P_{b_j^R}, P_{Ab_j^R}]=0$, by Lemma~\ref{lemma:lemma0}, $[P_{AB}, P_{BC}]=0$. 
\end{proof}
\noindent
Using Proposition~\ref{proposition:commutation_markov}, with a judicious choice of subsystems used in defining the parent Hamiltonian [Eq.~\eqref{eq:hamiltonian_global}], one can show that all the terms in Eq.~\eqref{eq:hamiltonian_global} commute with each other. 

However, it is a priori not obvious whether ensuring such commutation relation \emph{while ensuring LTQO} is possible. We now turn to this issue in Section~\ref{sec:parent_hamiltonian_final}. An explicit construction that satisfies both shall be given, which establishes that there exists a local parent Hamiltonian for the reference state $\sigma$ which moreover has a stable spectral gap.

\section{Constructing the parent Hamiltonian}
\label{sec:parent_hamiltonian_final}

\begin{figure}[t]
    \centering
    \begin{tikzpicture}
        \node[] () at (0,0) {\includegraphics[width=0.75\textwidth]{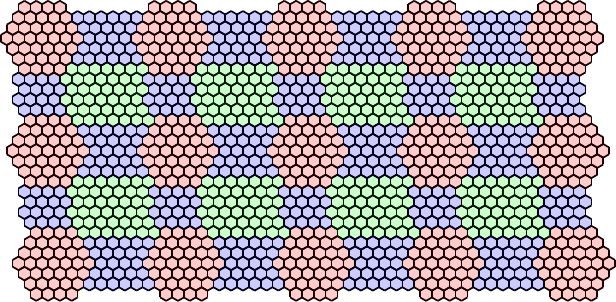}};
        \draw[dotted, line width=2pt]  (-7, 0) -- ++ (-0.4, 0);
        \draw[dotted, line width=2pt]  (7, 0) -- ++ (0.4, 0);
        \draw[dotted, line width=2pt]  (0, 3.5) -- ++ (0, 0.4);
        \draw[dotted, line width=2pt]  (0, -3.5) -- ++ (0, -0.4);
    \end{tikzpicture}
    \caption{Decomposition of the plane into $0$- $1$-, and $2$-cells, corresponding to the red, blue, and green balls.}\label{fig:coarse_grained_v2}
\end{figure}
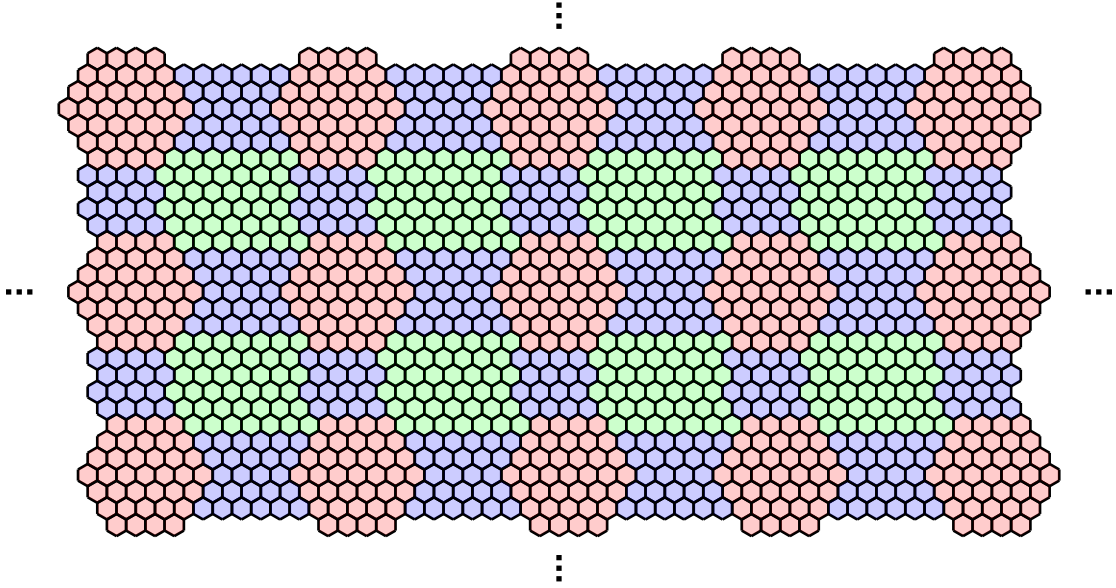

In Section~\ref{sec:ltqo}, we showed that the Hamiltonian in Eq.~\eqref{eq:hamiltonian_global} satisfies the LTQO condition. The key assumption was that the set of subsystems $\mathcal{S}$ in Eq.~\eqref{eq:hamiltonian_global} is dense enough so that for every elementary disk $D$, there exists some $X\subset \mathcal{S}$ such that $D\cup N(D)\subset X$. For \emph{any} $\mathcal{S}$ that satisfies this condition, the LTQO condition follows. In Section~\ref{sec:commuting_projectors}, we established a sufficient condition under which two terms in the Hamiltonian [Eq.~\eqref{eq:hamiltonian_global}] commute. This is possible if for any $X, Y\subset \mathcal{S}$, there exists a set of subsystems $A, B,$ and $C$ such that $X=AB$, $Y=BC$, and $I(A:C|B)_{\sigma}=0$ [Proposition~\ref{proposition:commutation_markov}]. What is not obvious is whether one can choose the set $\mathcal{S}$ such that the LTQO condition is satisfied while ensuring that the local terms commute with each other. The main purpose of this Section is to show that this is possible by constructing an explicit example. (Such a set $\mathcal{S}$ is not unique, and we provide another valid choice in Appendix~\ref{appendix:another-H}.) The Hamiltonian defined this way will have, by Ref.~\cite{Michalakis2013}, a gap stable against any extensive but weak enough perturbations.

Our construction is based on a decomposition of a two-dimensional plane into $0-$, $1-$, and $2-$cells. We remark that, while we name these objects as cells of different dimensions, they are nevertheless all balls on the plane. In Fig.~\ref{fig:coarse_grained_v2}, the $0$-, $1$-, and $2$-cells are denoted in red, blue, and green, respectively. A crucial property that we insist is that each $i$-cell is separated by other $i$-cells by some nonzero distance. We then define a cell $c_i$, associated with a $2$-cell, in the following way. For a given $2$-cell, collect all the $0$- and $1$-cells that neighbor this $2$-cell, and take a union. We refer to this union as $c_i$, also colloquially referred to as a cell in this paper. Our Hamiltonian is defined as 
\begin{equation}
    H = \sum_{i: 2-\text{cells}} (I-P_{c_i}). \label{eq:local_hamiltonian_candidate}
\end{equation}

There are two important properties of our decomposition. First, this decomposition ensures that the local terms in Eq.~\eqref{eq:local_hamiltonian_candidate} commute. For any pair $(c_i, c_j)$, if $c_i$ and $c_j$ have trivial intersection, $P_{c_i}$ and $P_{c_j}$ trivially commute. When $c_i$ and $c_j$ do overlap, they overlap in such a way that the following condition is satisfied:
\begin{equation}
    I(c_i \setminus c_j: c_j\setminus c_i| c_i \cap c_j)_{\sigma}=0. \label{eq:conditional_independence_for_commutation}
\end{equation}
(We shall prove this in Section~\ref{sec:conditional_independence_for_commutation}.) Setting $AB= c_i$ and $BC= c_j$, we can invoke Proposition~\ref{proposition:commutation_markov}, concluding that the projectors $P_{c_i}$ and $P_{c_j}$ commute. Thus, Eq.~\eqref{eq:local_hamiltonian_candidate} is a sum of local commuting projectors.

Second, for every elementary disk $D$, there exists at least one $c_i$ such that $D\cup N(D) \subset c_i$. (This is a simple fact that can be verified from Fig.~\ref{fig:coarse_grained_v2}.) 
Importantly, this implies that the Hamiltonian Eq.~\eqref{eq:local_hamiltonian_candidate} satisfies the LTQO condition [Proposition~\ref{proposition:ltqo}]. 

Thus, our decomposition serves two purposes: (i) ensuring that the local terms Eq.~\eqref{eq:local_hamiltonian_candidate} commute with each other and (ii) the same Hamiltonian satisfies the LTQO condition. By the well-known sufficient condition for the gap stability~\cite{Michalakis2013}, we can thus conclude that Eq.~\eqref{eq:local_hamiltonian_candidate} has a \emph{stable} spectral gap.

\subsection{Conditional independence}
\label{sec:conditional_independence_for_commutation}
In this Section, we seek to prove Eq.~\eqref{eq:conditional_independence_for_commutation}. Note that, when $c_i$ and $c_j$ have nontrivial intersection, there are three different ways in which this can happen. The examples are shown in Fig.~\ref{fig:overlap_regions}, with the convention that the red, green, and the blue regions represent $c_i\setminus c_j$, $c_i \cap c_j$, and $c_j\setminus c_i$, respectively.

\begin{figure}[h]
    \centering
    \begin{subfigure}[b]{0.3\textwidth}
        \centering
        \includegraphics[width=\textwidth]{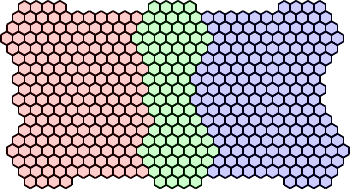}
        \caption{Horizontal overlap}
    \end{subfigure}
    \begin{subfigure}[b]{0.21\textwidth}
        \centering
        \includegraphics[width=\textwidth]{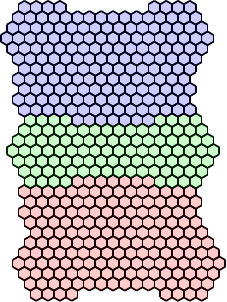}
        \caption{Vertical overlap}
    \end{subfigure}
    \begin{subfigure}[b]{0.3\textwidth}
        \centering
        \includegraphics[width=\textwidth]{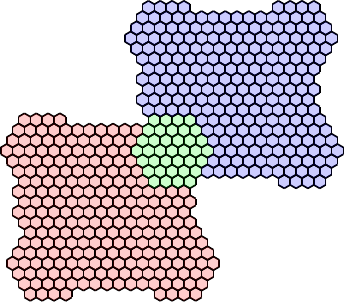}
        \caption{Diagonal overlap}
    \end{subfigure}
    \caption{There are three different ways in which $c_i$ and $c_j$ can overlap. Here red, green, and blue regions represent $c_i\setminus c_j$, $c_i \cap c_j$, and $c_j\setminus c_i$, respectively.}
    \label{fig:overlap_regions}
\end{figure}

In all three cases, we can prove the requisite conditional independence conditions [Eq.~\eqref{eq:conditional_independence_for_commutation}] from an extension of the axiom \textbf{A1} [Section~\ref{sec:extending_a1}]. For instance, suppose we seek to prove 
\begin{equation}
    I(A:C|B)_{\sigma}=0 \quad \text{ for } \quad A, B, C \text{ in Fig.~\ref{fig:cmi_upper_bounds}(a),}
\end{equation}
setting $A= c_i\setminus c_j$, $B= c_i \cap c_j$, and $C=c_j\setminus c_i$. (A note on convention: In Fig.~\ref{fig:cmi_upper_bounds}, the subsystem $A, B, C,$ and $D$ are colored in red, green, blue, and yellow.) We can use 
\begin{equation}
    I(A:C|B)_{\sigma}\leq S(C|B)_{\sigma} + S(C|D)_{\sigma} \quad \text{ for } A, B, C, D \quad \text{in Fig.~\ref{fig:cmi_upper_bounds}(a)}
\end{equation}
The analogous bounds for the other cases uses the choice of $A, B, C,$ and $D$ shown in Fig.~\ref{fig:cmi_upper_bounds}(b) and (c). 

\begin{figure}[h]
    \centering
    \begin{subfigure}[b]{0.3\textwidth}
        \centering
        \includegraphics[width=\textwidth]{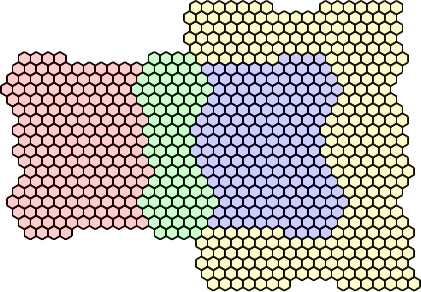}
        \caption{Horizontal overlap}
    \end{subfigure}
    \begin{subfigure}[b]{0.275\textwidth}
        \centering
        \includegraphics[width=\textwidth]{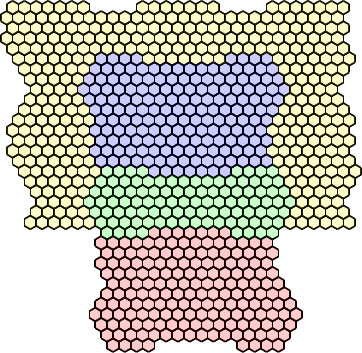}
        \caption{Vertical overlap}
    \end{subfigure}
    \begin{subfigure}[b]{0.3\textwidth}
        \centering
        \includegraphics[width=\textwidth]{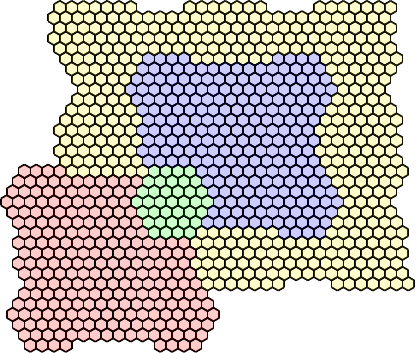}
        \caption{Diagonal overlap}
    \end{subfigure}
    \caption{The choice of subsystems used for upper bounding conditional mutual information. Here red, green, blue, and yellow regions represent $A= c_i \setminus c_j, B=c_i\cap c_j, C= c_j \setminus c_i,$ and $D$.}
    \label{fig:cmi_upper_bounds}
\end{figure}

The key point is that, in all these cases, one can show $S(C|B)_{\sigma} + S(C|D)_{\sigma}=0$; this immediately implies Eq.~\eqref{eq:conditional_independence_for_commutation}. So the natural question is how one can prove $S(C|B)_{\sigma} + S(C|D)_{\sigma}=0$ for the subsystems shown in Fig.~\ref{fig:cmi_upper_bounds}. This follows straightforwardly  through a simple but repetitive applications of the argument presented in Section~\ref{sec:extending_a1}. 

\section{Gapped domain wall}
\label{sec:gapped_domain_wall}
Our discussion so far pertained to the construction of the parent Hamiltonian of the states satisfying the axioms \textbf{A0} and \textbf{A1} everywhere. However, there is a class of states for which \textbf{A1} is violated along a line. In Ref.~\cite{Shi2021}, some of us advocated viewing such a line as a (codimension-$1$) defect, also known as the \emph{gapped domain wall}. Despite its name, it was previously unknown if the state in which a gapped domain wall is present is indeed stabilized by a gapped parent Hamiltonian. This is what we intend to prove in this Section.  

In the language of Ref.~\cite{Shi2021}, a gapped domain wall can be viewed as a one-dimensional line along which the axiom \textbf{A1} may be violated. More precisely, let us consider a partition of a 2D plane into hexagonal cells, consisting of the ones that lie at the upper/lower half with respect to a line [Fig.~\ref{fig:gdw_setup}]. While \textbf{A0} continues to hold in the vicinity of every face, and \textbf{A1} continues to hold in the vicinity of every face that does not touch the domain wall, \textbf{A1} \emph{is} modified for the faces that are in contact with the domain wall [Eq.~\eqref{eq:a1_boundary_valid},~\eqref{eq:a1_boundary_invalid}]. To repeat what we said in Section~\ref{sec:summary}, \textbf{A1} may not hold if the boundaries between $B$ and $D$ are both away from the domain wall and lie on the same side with respect to the domain wall. Otherwise, the \textbf{A1} continues to hold~\cite{Shi2021}. 

We remark that, just like the axioms can be extended in the absence of the gapped domain wall, the axioms for the gapped domain wall can also be extended to larger regions. The argument itself is identical to the one presented in Section~\ref{sec:extensions_of_axioms}, with details presented in Ref.~\cite{Shi2021}. As such, we omit the proof.

The main question we ask at this point is whether one can still show the existence of a gapped parent Hamiltonian with a stable spectral gap. We claim that this is indeed the case. The idea behind the proof is similar to the one presented in Section~\ref{sec:parent_hamiltonian_final}, accompanied by a modified argument for the subsystems that pass through the domain wall. More specifically, we propose to use Eq.~\eqref{eq:hamiltonian_global} as the parent Hamiltonian. The location of the domain wall with respect to the $0$-, $1$-, and $2$-cells are shown in Fig.~\ref{fig:decomposition_gdw}.

\begin{figure}[h]
    \centering
    \includegraphics[width=0.75\textwidth]{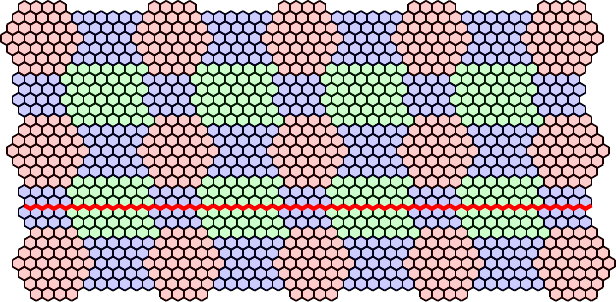}
    \caption{Location of the gapped domain wall (red line) with respect to the $0$-, $1$-, and $2$-cells, corresponding to the red, blue, and green balls.}
    \label{fig:decomposition_gdw}
\end{figure}

Now the task at hand is to prove that the Hamiltonian Eq.~\eqref{eq:hamiltonian_global}, with respect to the decomposition in Fig.~\ref{fig:decomposition_gdw}, has a stable spectral gap. As before, we can conclude so if we can prove the LTQO condition and the local gap condition. 

For proving the LTQO condition, we follow the argument similar to the one in Section~\ref{sec:ltqo}. Consider an arbitrary disk $D$ which, if it intersects with the domain wall, the intersection is a single interval.\footnote{Otherwise, the ensuing argument does not work. For instance, one can consider a disk $D$ whose intersection with the domain wall consists of two intervals. For such a disk, the reduced density matrix of an element of $\text{ker} H_D$ over $D'$ may be different from $\sigma_{D'}$~\cite{Shi2021}.} Also, consider a disk $D'\subset D$ consisting of faces in $D$ that are separated from $\partial D$ by at least some positive distance, such that over any elementary disk associated with a face in $D',$ any reduced density matrix over this elementary disk over $\ker H_D$ is indistinguishable from that of $\sigma$. By the discussion in Section~\ref{sec:ltqo}, such a $D'$ exists. The key question is whether $\rho_{D'}=\sigma_{D'}$ for any $\rho_{D'}$, which is a reduced density matrix of an element of $\ker H_D$ over $D'$. This follows from Lemma~\ref{lemma:a0_consequence}, by using \textbf{A0}. The proof of Proposition~\ref{proposition:ltqo} applies verbatim at this point, establishing the LTQO condition. 

For proving the local gap condition, it should be obvious that the terms that are away from the domain wall trivially commute with each other by the discussion in Section~\ref{sec:commuting_projectors}. What remains is to prove that even the terms that touch the domain wall commute with each other. To prove that, by Proposition~\ref{proposition:commutation_markov}, it suffices to prove $S(B|C)_{\sigma} + S(B|D)_{\sigma}=0$ for the subsystems shown in Fig.~\ref{fig:cmi_upper_bounds_gdw}. Luckily, these are all facts that follow from the extension of the axioms~\cite{Shi2021}. Therefore, the requisite conditional independence condition that ensures the commutativity of the local terms of the Hamiltonian [Proposition~\ref{proposition:commutation_markov}] follows.

\begin{figure}[h]
    \centering
    \begin{subfigure}[b]{0.3\textwidth}
        \centering
        \includegraphics[width=\textwidth]{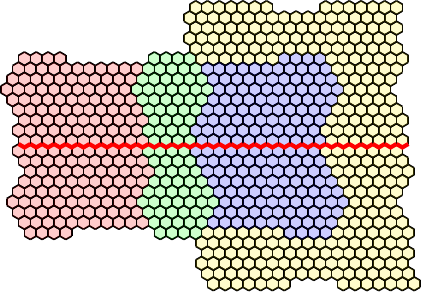}
        \caption{Horizontal overlap}
    \end{subfigure}
    \begin{subfigure}[b]{0.275\textwidth}
        \centering
        \includegraphics[width=\textwidth]{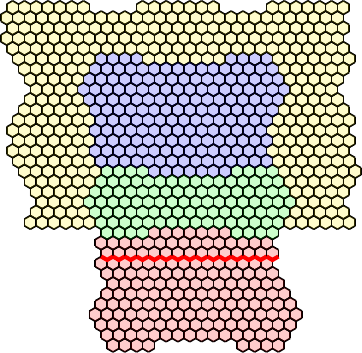}
        \caption{Vertical overlap}
    \end{subfigure}
    \begin{subfigure}[b]{0.3\textwidth}
        \centering
        \includegraphics[width=\textwidth]{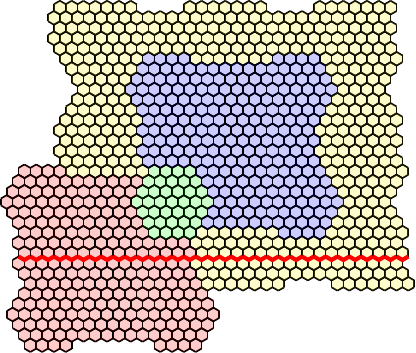}
        \caption{Diagonal overlap}
    \end{subfigure}
    \caption{The choice of subsystems used for upper bounding conditional mutual information in the presence of a gapped domain wall. Here, red, green, blue, and yellow regions represent $A= c_i \setminus c_j, B=c_i\cap c_j, C= c_j \setminus c_i,$ and $D$.}
    \label{fig:cmi_upper_bounds_gdw}
\end{figure}

To summarize, both the LTQO and the local gap condition continue to be satisfied for the Hamiltonian Eq.~\eqref{eq:hamiltonian_global}, even in the presence of gapped domain walls. Thus, we conclude that this Hamiltonian has a stable spectral gap, justifying the name ``gapped domain wall.''

\section{Weight reduction}
\label{sec:weight_reduction}

While the construction of the parent Hamiltonian in Section~\ref{sec:parent_hamiltonian_final} and~\ref{sec:gapped_domain_wall} are local in the sense that every term in the Hamiltonian acts nontrivially on a ball of $O(1)$ size, that constant is rather large. One might wonder if it is possible to reduce each term's \emph{weight}, defined as the number of faces contained in the support of the term.

To that end, we introduce a simple fact that lets us reduce the weight all the way down to three. 
\begin{lemma}
    \label{lemma:projector_decomposition}
    Let $\rho_{ABC}$ be a state defined on a finite dimensional Hilbert space $\mathcal{H}_A \otimes \mathcal{H}_B \otimes \mathcal{H}_C$ that satisfies $I(A:C|B)_{\rho}=0$. Let $P_{AB}$, $P_{BC}$, and $P_{ABC}$ be the projector onto the support of $\rho_{AB}$, $\rho_{BC}$, and $\rho_{ABC}$ respectively. Then
    \begin{equation}\label{eq:PP=P}
        P_{ABC} = P_{AB} P_{BC}.
    \end{equation}
\end{lemma}
\begin{proof}
The proof follows immediately from the fact that $P_{b_j^L} P_{Ab_j^L} = P_{Ab_j^L}$, $P_{b_j^R} P_{b_j^RC} = P_{b_j^RC}$, and Eq.~\eqref{eq:projectors_decomposition}. 
\end{proof}

From Lemma~\ref{lemma:projector_decomposition}, it follows that\footnote{Explicitly,  $Q\equiv 1-P$ is a projector. Eq.~\eqref{eq:PP=P} implies $1-Q_{ABC}=(1-Q_{AB})(1-Q_{BC})=(1-Q_{BC})(1-Q_{AB})$. Thus $Q_{AB} + Q_{BC} = Q_{ABC}+Q_{AB}Q_{BC}\ge Q_{ABC}$. The last inequality needs the fact that $Q_{AB}Q_{BC}$ is another projector, which follows from $Q_{AB}Q_{BC}= Q_{BC} Q_{AB}$.} 
\begin{equation}
    (I-P_{AB}) + (I-P_{BC}) \geq I- P_{ABC}. \label{eq:inequality_projectors}
\end{equation}
Using Eq.~\eqref{eq:inequality_projectors} recursively, we can obtain a Hamiltonian $H'$ such that $H'\geq H$, where $H$ is the Hamiltonian defined in Section~\ref{sec:parent_hamiltonian_final} or in Section~\ref{sec:gapped_domain_wall} and $H'$ consists of terms of weight at most three. 

Importantly, note that $\text{ker}(H) = \text{ker}(H')$. This is because the kernel of the left hand side and the right hand side of Eq.~\eqref{eq:inequality_projectors} are identical due to Lemma~\ref{lemma:projector_decomposition}. Thus $H'$ is also a parent Hamiltonian of the reference state that satisfies the local topological quantum order condition. Repeatedly applying Eq.~\eqref{eq:inequality_projectors}, one can reduce the weight of the local terms all the way down to three. However, one disadvantage of this Hamiltonian is that the local terms may not commute with each other.

\section{Discussion}
\label{sec:discussion}

In this paper, we constructed parent Hamiltonians for the states that satisfy the entanglement bootstrap axioms~\cite{shi2020fusion,Shi2021}. This parent Hamiltonian is a local commuting projector Hamiltonian that satisfies the local topological quantum order condition. Therefore, this Hamiltonian has a stable spectral gap. Thus, our result establishes that the states that obey the entanglement bootstrap axiom define a stable phase of matter.

The main novel technical insight is Proposition~\ref{proposition:commutation_markov}, which relates the conditional independence of the ground state to the commutativity of the projectors onto the supports of different local reduced density matrices. Armed with this fact, constructing a parent Hamiltonian reduces to the problem of finding an appropriate decomposition of the lattice that satisfies a certain combinatorial property [Section~\ref{sec:parent_hamiltonian_final}].

We expect this insight to be useful for a number of applications, which we summarize below. First, it is likely that our construction continues to work in higher dimensions, such as in three spatial dimensions~\cite{Huang2023,Shi2023}. One may also consider parent Hamiltonians for states that host topological excitations. It may be possible to also define a ``refined'' parent Hamiltonian whose excitations label the anyon excitations. We leave such studies for future work. 

We also remark that there is an interesting tension between our work and the formula for the chiral central charge~\cite{Kim2022,Kim2022a}. This formula was derived under \textbf{A1}, modulo an extra physical assumption. However, what we showed was that, \textbf{A1}, together with \textbf{A0}, implies that there ought to be a parent Hamiltonian which is a local commuting projector. Such a system cannot have any edge energy current~\cite{Kane1997,kitaev2006anyons}, and as such, would be in tension with any nonzero value of the chiral central charge. More precisely, our result suggests that the entanglement bootstrap axioms cannot be satisfied \emph{exactly} for ground states of systems that have nonzero chiral central charge, at least if each site is described by a finite-dimensional Hilbert space. Physically, this is because such a Hamiltonian necessarily has a vanishing energy current at the edge at all temperatures, which is possible only if the chiral central charge is zero. 

There is, nonetheless, a lingering question concerning the recently proposed modular commutator formula \cite{Kim2022} for chiral central charge $c_-$: 
\begin{equation}
    J(A,B,C)= \frac{\pi}{3} c_-, \nonumber
\end{equation}
where $A,B,C$ is a partition of a disk into three charts meeting at a point. Do we always have $J(A,B,C)=0$, under exact {\bf A0} and {\bf A1}? We leave it as an open problem.

Our work raises several natural questions, which we leave for future work. First, we note that having a parent Hamiltonian with the local TQO condition is a property invariant under a constant-depth circuit, whereas the axiom \textbf{A1} is not. Therefore, the existence of such a parent Hamiltonian is a strictly weaker assumption that holds for all states in the same phase as a state that satisfies \textbf{A0} and \textbf{A1}.\footnote{More precisely, while the local TQO condition is typically assumed only for square-shaped regions~\cite{Bravyi2010} or (geometric) balls~\cite{Michalakis2013}, we may demand this condition for any region that is \emph{topologically} a ball. Simply demanding local TQO for square-shaped regions or geometric balls would not reproduce Ref.~\cite{shi2020fusion}; such a condition is satisfied even for systems with gapped boundaries/domain walls, wherein the anyon contents are different on each side of the domain wall.} What can you prove from this weaker assumption? In particular, can one reproduce the results in Ref.~\cite{shi2020fusion}?

Second, it will be interesting to understand if our approach for constructing a commuting parent Hamiltonian works even if the local Hilbert space dimension is infinite. In fact, we think there are physical reasons to believe this is impossible, as we explain below. Note that our proof relied on the structure theorem for quantum Markov chains~\cite{Hayden2004}. To the best of our knowledge, an infinite-dimensional analog of this result is not currently known. If such a result can be proven and if one accepts that the entanglement bootstrap axioms can be satisfied even if the chiral central charge is nonzero (by having infinite local Hilbert space dimension), we are led to the conclusion that there can be a commuting parent Hamiltonian, even for systems with nonzero chiral central charge! However, this seems unlikely because the energy current produced by such Hamiltonian must necessarily be zero. On this physical ground, one may expect that the structure theorem for quantum Markov chains breaks down in infinite dimensions. Making this argument more precise or coming up with a rigorous counterexample is left for future work.

\section*{Acknowledgement}
IK acknowledges supports from NSF under award number PHY-2337931. DR acknowledges funding from NTT (Grant AGMT DTD 9/24/20). BS was supported by the Simons Collaboration on Ultra-Quantum Matter, a grant from the Simons Foundation (652264 JM). TCL is supported in part by funds provided by the U.S. Department of Energy (D.O.E.) under the cooperative research agreement DE-SC0009919 and by the Simons Collaboration on Ultra-Quantum Matter, which is a grant from the Simons Foundation (652264 JM).

\appendix

\section{Another explicit parent Hamiltonian}\label{appendix:another-H}

We constructed a commuting parent Hamiltonian Eq.~\eqref{eq:hamiltonian_global} satisfying the local topological quantum order condition. Such a construction rests on the existence of a set $\mathcal{S}$ with a couple of special properties. In Section~\ref{sec:parent_hamiltonian_final}, we explicitly constructed a valid such set $\mathcal{S}$. We emphasize that the specific form of the regions in the set $\mathcal{S}$ is not important. What is important is the following set of properties.

\begin{definition}[Markov cover] \label{def:Markov-Cover}
We say the set of subsystems $\mathcal{S}$ a Markov cover if it satisfies
    \begin{enumerate}
    \item (Markov condition) For  $\forall X, Y \in \mathcal{S}$ such that $X\cap Y \ne \emptyset $, The conditional mutual information of the reference state $I(X\setminus Y, Y\setminus X | X\cap Y)_\sigma =0$.
    \item (cover condition) For every elementary disk $D$, there exists some $X\in\mathcal{S}$ such that $D\cup N(D) \subset X$.  
\end{enumerate}
\end{definition}

As one may expect, there are many alternative choices. In this Appendix, we provide another valid set $\mathcal{S}$ for illustration purposes. Consider a set of red hexagons on the hexagon lattice [Fig.~\ref{fig:space-simple-color}].
\begin{figure}[h]
\centering
\includegraphics[width=10cm]{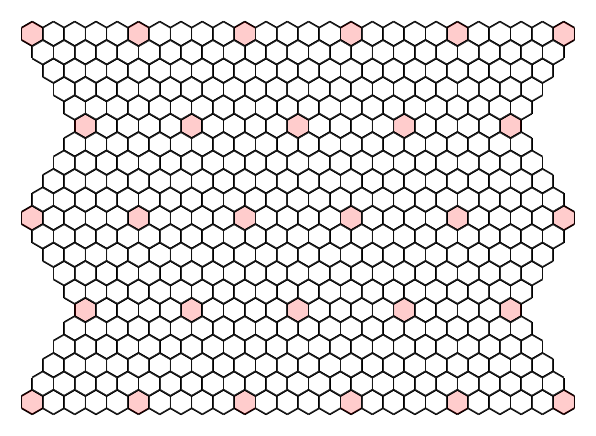}
\caption{Partition used for a valid set $\mathcal{S}$. \label{fig:space-simple-color}}
\end{figure}
Note that the distance between nearby red hexagons is equal to five.

Let us consider the following set: 
\begin{equation}\label{eq:S-candidate-2}
    \mathcal{S}= \{X | X = \textrm{a red face } f_{\rm red} \cup \textrm{faces within distance 5 to $f_{\rm red}$ that are not red}  \}.
\end{equation}  
It immediately follows that every $X \in \mathcal{S}$ is identical in its shape, shown below:
\begin{equation}\label{eq:red-only}
\begin{tikzpicture}[baseline={([yshift=-.55ex]current bounding box.center)}]
\node[] () at (0,0) {\includegraphics[width=3cm]{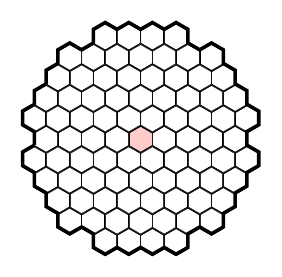}};
\end{tikzpicture}.
\end{equation}

As a matter of fact, the set $\mathcal{S}$ in Eq.~\eqref{eq:S-candidate-2} is a Markov cover. We need to check two properties, which explain below. First, the set $\mathcal{S}$ satisfies the Markov property; see the first item in Definition~\ref{def:Markov-Cover}. The relative positions between two regions $X, Y \in \mathcal{S}$, up to rotations, are the following 
    \begin{equation}\label{eq:red-only-overlap}
\begin{tikzpicture}[baseline={([yshift=-.55ex]current bounding box.center)}]
\node[] () at (0,0) {\includegraphics[width=10cm]{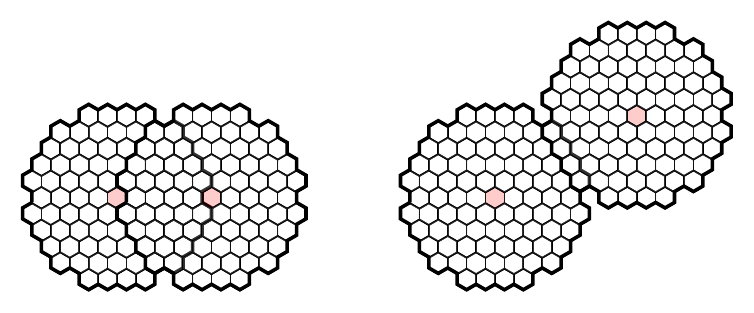}};
\end{tikzpicture}.
\end{equation}
Thus, the Markov property can be checked explicitly, as we have done in Fig.~\ref{fig:overlap_regions} and \ref{fig:cmi_upper_bounds}. Second, we would like to check the cover condition (i.e., the second item of Definition~\ref{def:Markov-Cover}). For any elementary disk $D$ centered at a face $f_D$, we find its nearest red face $f'_{\rm red}$.\footnote{If such a choice is not unique, pick one.} From the construction, we see that the distance between $f_D$ and $f'_{\rm red}$ can only be 0, 1, 2, 3. We end up with a small set of cases to check. In all such cases, indeed $D\cup N(D)$ is contained in the cover $X$ centered at the red face $f'_{\rm red}$.  Thus, the Hamiltonian Eq.~\eqref{eq:hamiltonian_global} is a commuting Hamiltonian that satisfies LTQO if we choose the set $\mathcal{S}$ as Eq.~\eqref{eq:S-candidate-2}.

 \subsection{Reference states with anyon excitations}
 
 Consider a known Hamiltonian with topological order and non-abelian anyons, such as a string-net model.  Say we take our reference state to be not the ground state, but rather some state with non-abelian anyons present.  Their presence can be diagnosed by violations of \textbf{A1}.  Can we still construct a corresponding parent Hamiltonian?  Or more generally, what happens if we consider a reference state with  \textbf{A1} violations in some regions?
 
 We comment on this briefly using the set $\mathcal{S}$ in Eq.~\eqref{eq:S-candidate-2}.
 The nontrivial case is to have a non-Abelian anyon in the center of a disk, i.e. at the red hexagon in \eqref{eq:red-only}. Suppose the quantum dimension of the anyon is $d_a$ ($d_a\ge \sqrt{2}$ for non-Abelian anyons), the constant sub-leading term in Eq.~\eqref{eq:strict-area} is modified by $\log d_a$ if disk $A$ contains an anyon  \cite{Kitaev2006}:
\begin{align} \label{eq:strict-area-da}
S(A)=\alpha |\partial A| - \gamma + \log d_a.
\end{align}
Suppose there are multiple non-Abelian anyons on the plane, each is centered at some $X \in \mathcal{S}$, then the parent Hamiltonian may have multiple ground states even on a topologically trivial manifold. Nonetheless, the steps leading to the stable energy gap (as in Section~\ref{sec:parent_hamiltonian_final}) still go through, provided that the anyons are well separated from each other. In particular, the resulting parent Hamiltonian Eq.~\eqref{eq:hamiltonian_global} is commuting as long as those $X_i,X_j \in \mathcal{S}$ containing anyons at the center are such that $X_i \cap X_j = \emptyset$. This is because, while \textbf{A1} may be violated, the conditional independence condition remains true. The additional contribution of $\log d_a$ to the entanglement entropy cancels out each other in the conditional mutual information.

\bibliographystyle{myhamsplain2}
\bibliography{bib}

\end{document}